\documentclass[11pt]{article}
\usepackage[english]{babel}
\usepackage[utf8]{inputenc}
\usepackage[T1]{fontenc}
\usepackage{authblk}
\usepackage{graphicx}
\usepackage{amsmath,mathrsfs,dsfont}
\usepackage{amsthm}
\usepackage[all]{xy}
\usepackage{graphics}
\usepackage{tensor}
\usepackage{verbatim}
\usepackage{nicefrac}
\usepackage{amsfonts}
\usepackage{amssymb,latexsym}
\usepackage{color}
\usepackage{mathrsfs,dsfont}
\usepackage{appendix}
\usepackage{slashed}
\usepackage{cite}
\usepackage{url}
\usepackage{hyperref}
\usepackage{comment}
\newtheorem{proposition}{Proposition}
\newtheorem{corollary}{Corollary}

\DeclareMathOperator{\Tr}{Tr}

\newcommand{\ba}{\begin{eqnarray}}
\newcommand{\ea}{\end{eqnarray}}

\DeclareMathOperator{\py}{\dot{\partial}}

\begin{document}

\title{Multimetric Finsler Geometry}
\author[1]{Patr\'{i}cia Carvalho\thanks{patriciacarvalho.fisica@gmail.com}}
\author[1]{Cristian Landri\thanks{cristianlandri@gmail.com}}
\author[1]{Ravi Mistry\thanks{ravi.mistry.r@gmail.com}}
\author[1,2]{Aleksandr Pinzul\thanks{apinzul@unb.br}}
\affil[1]{Universidade de Bras\'{\i}lia, Instituto de F\'{\i}sica 70910-900, Bras\'{\i}lia, DF, Brasil}
\affil[2]{International Center of Physics C.P. 04667, Brasilia, DF, Brazil}
%\date{}
\maketitle

\begin{abstract}
Motivated in part by the bi-gravity approach to massive gravity, we introduce and study the multimetric Finsler geometry. For the case of an arbitrary number of dimensions, we study some general properties of the geometry in terms of its Riemannian ingredients, while in the 2-dimensional case, we derive all the Cartan equations as well as explicitly find the Holmes-Thompson measure.
\end{abstract}
\section{Introduction}

Finsler geometry is a natural generalization of the Riemannian one. Roughly speaking, it is based on the notion of an interval without the quadratic restriction, i.e. on the interval that locally does not look like the one given by the Pythagoras theorem. In fact, this possibility was already discussed by Riemann in his habilitation lecture in 1854. The general study was undertaken only 50 years later by P. Finsler. Since then, thanks to many contributions by various mathematicians, Finsler geometry has developed into a separate branch of mathematics that includes the standard Riemannian geometry as a special case {\cite{Anastasiei,Antonelli,Bao,Chern,Mo}}. At the same time, Finsler geometry has been found to be useful in many applications ranging from biology {\cite{Antonelli1993OnYC,ANTONELLI2005899}} and cosmology {\cite{Huang:2007en,Chang:2008yv,Silva:2016qkj}} to Lorentz violating models \cite{Alan_Kosteleck_2012} and non-linear optics {\cite{Roman,Miron}}, among others.

Our interest in Finsler geometry comes from its natural appearance in the bi-metric formulation of the recent models for massive gravity {\cite{deRham:2014zqa}}. Namely, if one considers two different metric tensors on the same manifold, $\alpha_{ij}$ and $\beta_{ij}$, then the natural action for a ``minimally'' coupled point-like particle is given by {\cite{Akrami:2014lja}}
\ba\label{action_massive}
S = \int ds_\alpha + \int ds_\beta \ ,
\ea
where $ds_\alpha = \sqrt{\alpha_{ij}\dot{x}^i \dot{x}^j} dt$ and $ds_\beta = \sqrt{\beta_{ij}\dot{x}^i \dot{x}^j} dt$ are the usual intervals for the corresponding metrics (we included possible dimensionless coefficients into the definitions of the metric tensors). It is obvious that unless $\alpha$ is proportional to $\beta$, (\ref{action_massive}) will not lead to geodesics for some Riemannian geometry. Instead, defining
\ba\label{FFF1}
\mathrm{F}_\alpha(x,\dot{x}) := \sqrt{\alpha_{ij}\dot{x}^i \dot{x}^j}\ ,\ \mathrm{F}_\beta(x,\dot{x}) := \sqrt{\beta_{ij}\dot{x}^i \dot{x}^j}\ \mathrm{and}\ \mathrm{F}(x,\dot{x}):= \mathrm{F}_\alpha(x,\dot{x}) + \mathrm{F}_\beta(x,\dot{x})\ ,
\ea
we can trivially write (\ref{action_massive}) as
\ba\label{action_Finsler}
S = \int \mathrm{F}(x,\dot{x}) dt \ ,
\ea
which is exactly the functional, controlling the geodesic distance in the Finsler geometry defined by the Finsler function $\mathrm{F}(x,\dot{x})$ given in (\ref{FFF1}).

Then a very natural question arises: Can one have a purely geometric description of the dynamics of the bi-metric gravity? By this we mean the following point. The Einstein-Hilbert action is the simplest and the most natural action describing the dynamics of Riemannian structure on a manifold. Can we interpret the bi-gravity action \cite{deRham:2014zqa}
\ba\label{bigrav}
S_{2-gr} &=& \frac{M_{\alpha}^2}{2} \int \mathrm{d}^4 x\sqrt{|\alpha |} R[\alpha] + \frac{M_{\beta}^2}{2} \int \mathrm{d}^4 x\sqrt{|\beta |} R[\beta] + \nonumber \\
&&+\frac{M_{\alpha}^2 m^2}{4} \int \mathrm{d}^4 x\sqrt{|\alpha |} \sum\limits_{n=0}^4 c_n e_n (\mathds{1} - \sqrt{\alpha^{-1} \beta})
\ea
as an action for some type of Finslerian gravity? (In (\ref{bigrav}), $M_{\alpha ,\beta}$ and $m$ are some characteristic scales, $e_n (\mathds{X})$ are the symmetric polynomials of a matrix $\mathds{X}$ and $c_n$ are some coupling constants.) In regard to this question, two comments are in order.

Firstly, the formulation of Finslerian gravity is far from being settled (see, \cite{Vacaru2007,Vacaru2008,Pfeifer2013,Pfeifer2019} for some efforts in this direction). Partly, this is due to the fact that Finsler geometry has much more freedom in constructing invariants, even after one imposes some natural conditions, like metric compatibility. As a consequence, only one (if any) formulation of Finslerian gravity could lead to (\ref{bigrav}).

Secondly, what is typically called Finsler geometry corresponds to a generalization of Riemannian geometry. Its Lorentzian counterpart is much less understood (see \cite{Pfeifer2013} for a discussion). This means that, if successful, we should expect to arrive at a Euclidean version of the bi-metric gravity. We do not see this as a problem. The similar situation happens in another approach to gravity, the spectral action approach \cite{Chamseddine:1996zu,Chamseddine:2008zj}. This approach is naturally formulated in the Euclidean setting, but because the final answer is given in terms of geometrical objects, at the end one can treat it as being defined on a space-time with the Lorentzian signature. For example, the spectral action formulation of the so-called Landau-Lifshitz gravity \cite{Horava:2009uw} was done in \cite{Pinzul:2010ct,Pinzul:2014mva,Lopes:2015bra,Pinzul:2016dwy}, where it was treated as a gravity on some generalized geometry.

To address the question of the geometric formulation of the bi-metric gravity, first one has to study the geometry given by the Finsler structure (\ref{FFF1}). This paper is intended as the initial effort in this direction. We study the geometry defined by the multimetric generalization of (\ref{FFF1}). We obtain some general properties of this geometry as well as study in more details the example of the 2-dimensional multimetric Finsler geometry. The plan of the paper is as follows. In Section \ref{Preliminaries} we give a minimalistic review of Finsler geometry and its structures used in the main part. Section \ref{Multimetric general} introduces the multimetric Finsler structure and studies some of its general properties. Next, in Section \ref{Multimetric 2d}, we specify to the 2-dimensional case. This allows us to get some further results, unavailable at the moment for the general case, e.g. we find the closed expression for a natural measure on the multimetric Finsler space. We conclude with the discussion of what should be done next as well as consider some possible applications. The extensive Appendix provides a more geometric approach to some elliptic integrals used in our study of measure.

\section{Preliminaries}\label{Preliminaries}

Here we collect some necessary facts and definitions in Finsler geometry. The discussion will be sketchy and not very rigorous. For the detailed account, see {\cite{Anastasiei,Antonelli,Bao,Chern,Mo}}.

For an $n$-dimensional manifold $\mathcal{M}$, Finsler geometry is defined, first, by specifying on the tangent bundle, $T\mathcal{M}$, a Finsler structure and, second, by choosing the non-linear connection. More specifically, Finsler structure is a map, $\mathrm{F}: T\mathcal{M}\rightarrow [0,\infty )$, that is $\mathcal{C}^{\infty}$ on the slit tangent bundle, $T\mathcal{M}\!\setminus\!\{0\}$, satisfying\label{definition_F}

i) positive 1-homogeneity: $\mathrm{F}(x,\lambda y)=\lambda\mathrm{F}(x, y)$ for any $\lambda >0$ (here $(x,y)$ are coordinates of some trivialization of $T\mathcal{M}$, $x$'s being the coordinates on $\mathcal{M}$, while $y$'s - the coordinates along fibers);

ii) strong convexity on $T\mathcal{M}\!\setminus\!\{0\}$: the matrix $[\py_i \py_j \mathrm{F}^2]$ is positively defined (here we defined $\py_i := \frac{\partial}{\partial y^i}$).
The latter condition allows us to define a metric that depends on both, $x$ and $y$,
\ba\label{metric_def}
g_{ij}:=\frac{1}{2}\py_i\py_j \mathrm{F}^2 \ .
\ea
Using the homogeneity of $\mathrm{F}$, it is trivial to see that $\mathrm{F} = \sqrt{g_{ij}y^i y^j}$.
Introducing
\ba\label{l_def}
l_i := \py_i \mathrm{F}\equiv \frac{g_{ij}y^j}{\mathrm{F}} \ ,
\ea
the metric (\ref{metric_def}) has a convenient ADM-like decomposition
\ba\label{ADM}
g_{ij} = l_i l_j + \mathrm{F}\py_i \py_j \mathrm{F} =:l_i l_j + h_{ij} \ ,
\ea
where $\mathrm{rank} [h]=n-1$ with $h_{ij}y^j = 0$. The vectors $l_i$ satisfy $l^i l_i =1$ (where raising/lowering of the indices is done with the help of the metric (\ref{metric_def}), so $l^i = \frac{y^i}{\mathrm{F}}$) and are related to $h_{ij}$ via
\ba\label{lhrelation}
\dot{\partial}_i l_j = \dot{\partial}_i \frac{g_{jk}y^k}{\mathrm{F}} = -\frac{1}{\mathrm{F}}l_i l_j + \frac{1}{\mathrm{F}}g_{ij} \equiv \frac{1}{\mathrm{F}}h_{ij} \ ,
\ea
where we used that $g_{jk}$ is 0-homogeneous with respect to $y$ and $\py_i g_{jk} = \py_k g_{ji}$, i.e. $y^k\py_i g_{jk} = 0$.

The choice of a non-linear connection on $T\mathcal{M}$ is done by specifying the horizontal subspace of $T\mathcal{M}$. In this way, while the vertical subspace is naturally spanned by $\{\py_i\}$, the basis for the horizontal one is given by $\{\delta_i := \partial_i - N^j_{\ i} \py_j\}$, where $\partial_i := \frac{\partial}{\partial x^i}$, as usual. To further specify the non-liner connection $N^j_{\ i}$, one requires that it respects the Finsler structure, i.e.
\ba\label{CartanN}
\delta_i \mathrm{F} = 0\ \ \mathrm{or}\ \ \partial_i \mathrm{F} = N^j_{\ i} \py_j \mathrm{F}\equiv l_j N^j_{\ i} \ .
\ea
For this choice, the connection is called Cartan non-linear connection and it is explicitly given by
\ba\label{Cartan_def}
N^i_{\ j} = \Gamma^i_{\ jk}y^k - C^i_{\ jk}\Gamma^k_{\ rs}y^r y^s \ ,
\ea
where $\Gamma^i_{\ jk}$ are the standard (in general, $y$-dependent) Christoffel symbols for the metric (\ref{metric_def}) and
\ba\label{Cartan_tens}
C_{ijk} := \frac{1}{2}\py_i g_{jk}\equiv \frac{1}{4}\py_i\py_j\py_k \mathrm{F}^2
\ea
is the Cartan tensor. A Finsler geometry will be actually a Riemannian one, i.e. given by $\mathrm{F} = \sqrt{g_{ij}y^i y^j}$ with $g_{ij}$ depending only on $x$, if and only if $C_{ijk}=0$.

Riemannian spaces are not the only special cases of the general Finsler geometry. Other interesting classes are Berwald and Landsberg spaces. With the help of the horizontal derivative, which is defined for a tensor $T_{i_1\cdots i_k}$ as
\ba\label{hor}
T_{i_1\cdots i_k|j}:=\delta_{j}T_{i_1\cdots i_k} - \sum\limits_{s=1}^{s=k} T_{i_1\cdots j_s \cdots i_k}{}^C\Gamma^{j_s}{}_{i_s j} \ ,
\ea
where the Chern connection, ${}^C\Gamma^{k}{}_{ij}$, is given by
\ba\label{Chern_connection}
{}^C\Gamma^{k}{}_{ij}=\frac{g^{ks}}{2}\left(\delta_{i}g_{sj}+\delta_{j}g_{si}-\delta_{s}g_{ij}\right)\ ,
\ea
these spaces are defined by the following conditions:
\ba\label{BerLan}
C_{ijk|s}=0\ &-& \ \mathrm{Berwald\ space}\ ,\label{Ber} \\
\dot{C}_{ijk}:=C_{ijk|s}y^s=0\ &-& \ \mathrm{Landsberg\ space}\label{Lan}\ .
\ea
The importance of such spaces is due to their similarities with Riemannian ones: Berwald space preserves the auto-parallel curves, i.e. geodesics \cite{Anastasiei}, while Landsberg one, the holonomy invariance, at least in the indicatrix (i.e., the region where $F=1$)\cite{Antonelli}. Note, that one has the following (proper) inclusions
\ba
Riemannian \subset Berwaldian \subset Landsbergian \subset general\ Finsler \ .\nonumber
\ea

In the 2-dimensional case, $n=2$, the rank of $h$ is 1, so one can write $h_{ij}=m_i m_j$, for some vector $m_i$. Then (\ref{ADM}) takes the form
\ba\label{2dmetric}
g_{ij} = l_i l_j + m_i m_j\ ,
\ea
where
\ba\label{lm}
l^i m_i = 0 \ ,\ m^i m_i = 1 \ .
\ea
(As usual, the raising/lowering of the indices is done with the metric (\ref{2dmetric}).) The pair $(l, m)$ is called the Berwald zweinbeins.

It is easy to see that
\ba\label{mi}
m_i = \epsilon_{ij}l^j\ \ \mathrm{and}\ \ m^i = \epsilon^{ij}l_j \ ,
\ea
where
\ba\label{epsilon}
\epsilon_{ij} = \sqrt{{g}} \varepsilon_{ij}\ \ ,\ \ \epsilon^{ij} = \frac{1}{\sqrt{{g}}} \varepsilon^{ij} \ ,
\ea
with $\varepsilon_{ij}$ being the usual ``flat'' Levi-Civita symbol, $\varepsilon_{12}=1$, and ${g}=\det(g_{ij})$. The relation (\ref{lhrelation}) takes the form of the relation between $l$ and $m$:
\ba\label{lmrelation}
\py_i l_j = \frac{1}{\mathrm{F}}m_i m_j \ .
\ea

Any 2-dimensional Finsler geometry is completely characterized by three (pseudo)scalars: $I$, $J$ and $K$. Roughly speaking, while $I$ controls the non-Riemannianity of the geometry, $J$ tells whether it is Landsberg or not. $K$ is just a scalar curvature of a natural covariant derivative. These scalars enter the Cartan equations on $S\mathcal{M}^*$ - the dual projective sphere bundle (for the extensive discussion, see \cite{Antonelli},\cite{Bao}, also see section \ref{Cartan_equations} for some notations and further discussion):
\ba\label{Cartan_eqs}
d\omega^1 = - I \omega^1 \wedge \omega^3 + \omega^2 \wedge \omega^3 \ ,\ d\omega^2 = - \omega^1 \wedge \omega^3\ ,\ d\omega^3 = K \omega^1 \wedge \omega^2 - J \omega^1 \wedge \omega^3 \ .
\ea

\section{Multimetric geometry: general case}\label{Multimetric general}

Here we want to introduce and study some properties of the multimetric Finsler geometry, which we define by the following Finsler structure
\ba\label{N_multimetric}
\mathrm{F} = \sum\limits_{\mu=1}^{N} \mathrm{F}^{(\mu)} \ ,
\ea
where each $\mathrm{F}^{(\mu)}$ is just a Riemannian Finsler structure, i.e. $\mathrm{F}^{(\mu)}(x,y)=\sqrt{\alpha^{(\mu)}_{ij}(x)y^i y^j}$ with $\alpha^{(\mu)}_{ij}(x)$ being positively definite. At the moment, no restriction on the dimension of the manifold $\mathcal{M}$ is made. In the next section, we will specify to the 2-dimensional case, $n=2$.

First of all, we must verify that (\ref{N_multimetric}) really defines a Finsler structure in the sense of the discussion on the page \pageref{definition_F}. The only non-trivial check is to see if (\ref{N_multimetric}) satisfies the strong convexity condition.
\begin{proposition}\label{PropA}
Let $\mathrm{F}^{(1)}$ and $\mathrm{F}^{(2)}$ be arbitrary (not necessarily Riemannian) Finsler structures. Then
\ba
\mathrm{F}:=\mathrm{F}^{(1)}+\mathrm{F}^{(2)}
\ea
is also a Finsler structure.
\end{proposition}
\begin{proof}
Clearly $\mathrm{F}\ : T\mathcal{M}\rightarrow [0,\infty )$ and it is of $\mathcal{C}^\infty$ on the split bundle and positively 1-homogeneous in $y$ if $\mathrm{F}^{(1,2)}$ are.

Using
\ba\label{lh2}
l_i = \py_i \mathrm{F} \equiv \sum\limits^2_{\mu=1} l^{(\mu)}_i  \ ,\ \ h_{ij} = \mathrm{F}\py_i\py_j\mathrm{F}\equiv \sum\limits^2_{\mu=1} \frac{\mathrm{F}}{\mathrm{F}^{(\mu)}} h^{(\mu)}_{ij}
\ea
and plugging this in (\ref{ADM}) we have
\ba\label{positive_def}
g_{ij} = \left( l^{(1)}_i + l^{(2)}_i\right)\left( l^{(1)}_j + l^{(2)}_j\right) + \sum\limits^2_{\mu=1} \frac{\mathrm{F}}{\mathrm{F}^{(\mu)}} h^{(\mu)}_{ij} \ .
\ea
Because each $h^{(\mu)}_{ij}$ is positively definite (and non-degenerate) on the orthogonal complement of $y$ (in the sense of $\mathbb{R}^n$) and $\left( l^{(1)}_i + l^{(2)}_i\right)y^i \equiv \mathrm{F}^{(1)}+\mathrm{F}^{(2)} >0 $ we have our proof.
\end{proof}
Clearly, by induction, a sum of any number of Finsler structures is again a Finsler structure. In particular, this guarantees that (\ref{N_multimetric}) is well-defined.

Let us trivially rewrite (\ref{N_multimetric}) in a very suggestive form:
\ba\label{N_multimetric1}
\mathrm{F} = \left( 2\sum\limits_{\mu<\nu}^{N} \mathrm{F}^{(\mu)}\mathrm{F}^{(\nu)}  + \sum\limits_{\mu=1}^{N} \mathrm{F}^{(\mu)2}\right)^{\frac{1}{2}} \ .
\ea
In the bi-metric case this becomes (with $\mathrm{F}^{(1)}= \sqrt{\alpha_{ij}(x)y^i y^j}\ ,\ \mathrm{F}^{(2)}= \sqrt{\beta_{ij}(x)y^i y^j}$, cf. (\ref{FFF1}))
\ba\label{bimetric}
\mathrm{F} = \left( \sqrt{4\alpha_{ij}(x)\beta_{km}(x)y^i y^j y^k y^m}  + (\alpha_{ij}(x)+\beta_{ij}(x))y^i y^j\right)^{\frac{1}{2}} \ .
\ea
This is nothing but a special case of the so-called generalized $4^{th}$-root Finsler structure, which in the general $m^{th}$-root case is given by \cite{Tayebi2011,Tayebi2012}
\ba\label{m_root}
\mathrm{F}_m = \left( \left(A_{i_1 ... i_m}(x)y^{i_1}... y^{i_m} \right)^{\frac{2}{m}}  + B_{ij}(x)y^i y^j\right)^{\frac{1}{2}} \ .
\ea
So, we can think of (\ref{N_multimetric}) as a natural, physically motivated, multimetric generalization of the generalized $4^{th}$-root Finsler geometry.

The representation (\ref{N_multimetric1}) is useful in answering the following natural questions:

1) What are the conditions on the Finsler  structure (\ref{N_multimetric}) (or, equivalently, (\ref{N_multimetric1})) for the geometry to be of a special type, Riemannian or Lansbergian?

2) Why don't we take as a multimetric Finsler structure the following, seemingly equivalently natural, structure:
\ba\label{N_multimetric2}
\tilde{\mathrm{F}}^2 = \sum\limits_{\mu<\nu}^{N} \mathrm{F}^{(\mu)}\mathrm{F}^{(\nu)}\ ?
\ea

To answer the first question, let us note that $\sum\limits_{\mu=1}^{N} \mathrm{F}^{(\mu)2}$ in (\ref{N_multimetric1}) is a homogeneous polynomial in $y$'s, while each $\mathrm{F}^{(\mu)}\mathrm{F}^{(\nu)} \equiv \sqrt{\alpha^{(\mu)}_{ij}(x)\alpha^{(\nu)}_{km}(x)y^i y^j y^k y^m}$ is, in general, an irrational function of $y$'s. To correspond to a Riemannian case, each such function should also be a homogeneous polynomial in $y$'s, which is possible if and only if it is reducible, i.e. $\alpha^{(\mu)}_{ij}(x)\alpha^{(\nu)}_{km}(x)y^i y^j y^k y^m = (c^{(\mu\nu)}_{ij}(x)y^i y^j )^2$ for some $c^{(\mu\nu)}_{ij}$ and for each $\mu , \nu$. By the uniqueness of the factorization of a polynomial, this is possible if and only if $\alpha^{(\mu)}_{ij}(x) \sim\alpha^{(\nu)}_{ij}(x) \sim c^{(\mu\nu)}_{ij}(x)$. So, we have established the following
\begin{proposition}\label{PropB}
The multimetric Finsler structure (\ref{N_multimetric}) will define a Riemannian geometry if and only if for all $\mu , \nu$ exist positive functions of $x$ only, $\phi^{(\mu)}(x)$, ($\phi^{(\mu)}(x) \equiv 1$) such that
\ba
\alpha^{(\mu)}_{ij}(x) = \phi^{(\mu)}(x)\alpha^{(1)}_{ij}(x)\ .
\ea
In this case the corresponding Riemannian metric is given by
\ba
g_{ij}(x) = \left(\sum\limits_{\mu =1}^N \phi^{(\mu)}(x)\right)\alpha^{(1)}_{ij}(x)\ .
\ea
\end{proposition}

To analyze a criterium for the multimetric Finsler geometry to be of the Landsbergian type, the representation again (\ref{N_multimetric1}) proves to be very useful. As we mentioned above, due to this representation, we can think of the multimetric geometry as a generalisation of the generalized $4^{th}$-root Finsler structure and one can repeat almost verbatim the analysis of \cite{Tabata} to obtain  the following result: Given a Finsler structure, $\mathrm{F}$, as in (\ref{N_multimetric1}), it is Landsbergian if and only if
\ba\label{Lansb1}
\left( \frac{1}{\mathrm{F}} \py_r \py_s \py_t \left( \mathrm{F}^{(\mu)} \mathrm{F}^{(\nu)} \right) \right)\!{}_{|j} y^j =0 \ \ \mathrm{for\ all}\ \mu ,\nu \ .
\ea
Here the horizontal derivative is defined as in (\ref{hor}). Opening the derivations with respect to $y$'s and using that $\mathrm{F}$ is horizontally constant, i.e. $\delta_i \mathrm{F}=0$, we arrive at
\ba
\left(\alpha^{(\mu)}_{i\{r}\alpha^{(\nu)}_{st\}}y^i+\alpha^{(\nu)}_{i\{r}\alpha^{(\mu)}_{st\}}y^i\right)\!{}_{|j} y^j=0
\ea
as a Landsbergian condition in terms of the Riemannian metrics (here $\{\ldots\}$ stands for the cyclic permutations).

The answer to the second question, namely why we don't use (\ref{N_multimetric2}) as a candidate for the multimetric metric is less precise. The point is that this would require more conditions on the Riemannian ingredients, $\alpha^{(\mu)}_{ij}(x)$, to have a well-defined Finsler geometry. For example, in the bi-metric case, (\ref{N_multimetric2}) will correspond to the $4^{th}$-root Finsler geometry introduced in {\cite{Shimada}} and it is known that it does not always lead to a positively definite metric. The role of the ``deformation'' by $B_{ij}(x)$ in (\ref{m_root}) is exactly to guarantee that the resulting structure is the Finsler one.

Because the Finsler structure (\ref{N_multimetric}) is given in terms of more familiar Riemannian structures, corresponding to $\mathrm{F}^{(\mu)}$ in (\ref{N_multimetric}), it would be very desirable to express all the relevant Finsler-geometric quantities in terms of the Riemannian ones (plus possible ``interaction'' terms). We postpone the detailed analysis for the future study and here only make some preliminary considerations.

Trivially repeating the analysis from Proposition \ref{PropA} and using the definitions, we find the following factorized relations
\ba\label{relations}
l_i = \sum\limits_{\mu =1}^{N} l^{(\mu)}_i,\ \ l^i = \frac{\mathrm{F}^{(\mu)}}{\mathrm{F}} l^{(\mu)i},\ \ \frac{1}{\mathrm{F}} h_{ij} = \sum\limits_{\mu =1}^{N} \frac{1}{\mathrm{F}^{(\mu)}} h^{(\mu)}_{ij},\ \ {\mathrm{F}} h^i_{\ j} = \sum\limits_{\mu =1}^{N} {\mathrm{F}^{(\mu)}} h^{(\mu)i}_{\ \ \ \ j}\ .
\ea
Note that in the last relation the index on the left hand side is raised by $g^{ij}$, while on the right hand side $\alpha^{(\mu)ij}$ is used. Then for the full metric we immediately have\footnote{Of course, this also could be directly obtained using $g_{ij} = \py_i (\mathrm{F}l_j)$, which is easily verified.\label{footnote1}}
\ba\label{metric_relations}
g_{ij}=\sum\limits_{\mu,\nu =1}^{N} l^{(\mu)}_i l^{(\nu)}_j + \sum\limits_{\mu =1}^{N} \frac{\mathrm{F}}{\mathrm{F}^{(\mu)}} h^{(\mu)}_{ij} \equiv \sum\limits_{\mu =1}^{N} \alpha^{(\mu)}_{ij}  + \sum\limits_{\mu \ne \nu}^{N} \left(l^{(\mu)}_i l^{(\nu)}_j +\frac{\mathrm{F}^{(\nu)}}{\mathrm{F}^{(\mu)}} h^{(\mu)}_{ij}\right) \ .
\ea
The latter form in (\ref{metric_relations}) allows us to interpret the multimetric case as the ``interaction'' between different Riemannian structures, $\alpha^{(\mu)}_{ij}$, put on the same manifold $\mathcal{M}$: while the first term describes a ``non-interacting'' direct sum of the Riemannian structures, the second term is exactly the ``interaction''. More rigorously, if we calculate the Cartan tensor (\ref{Cartan_tens}) corresponding to the multimetric geometry, we see that only the second term leads to a non-trivial contribution. Because, as it was discussed above, this tensor is, in a sense, a measure of the non-Riemannianity of geometry, we can treat the second term in (\ref{metric_relations}) as interaction.

The expression for the nonlinear connection (\ref{Cartan_def}) is easily obtained with the help of the defining relation (\ref{CartanN}). From (\ref{CartanN}) we have
\ba\label{rel_nonlin1}
l_j N^j_{\ i} = \sum\limits_{\mu=1}^{N} l^{(\mu)}_j N^{(\mu)j}_{\ \ \ \,\ i} \ ,
\ea
which immediately gives (see the footnote \ref{footnote1})
\ba\label{rel_nonlin2}
g_{kj} N^j_{\ i} + \mathrm{F} l_j \py_k N^j_{\ i} = \sum\limits_{\mu=1}^{N} \left( \py_k (\mathrm{F} l^{(\mu)}_j ) N^{(\mu)j}_{\ \ \ \,\ i}  + \mathrm{F} l^{(\mu)}_j \py_k N^{(\mu)j}_{\ \ \ \,\ i}\right) \ .
\ea
This could be written more conveniently in terms of the so-called spray coefficients $G^i = N^i_{\ j} y^j$ or, inverting, $N^{i}_{\ j} = \frac{1}{2}\dot{\partial}_j G^{i}$. This is equivalent to the following identity
\ba\label{ydN}
y^j \dot{\partial}_k N^{i}_{\ j} = N^{i}_{\ k} \ .
\ea
(Proof: Just apply the Euler's theorem for homogeneous functions to $N^{i}_{\ j} = \frac{1}{2}\dot{\partial}_j G^{i}$, where $G^{i}$ is 2-homogeneous (with respect to $y$).) Using (\ref{ydN}) in (\ref{rel_nonlin2}) and applying (\ref{rel_nonlin1}) one more time, we get
\ba\label{rel_nonlin3}
g_{kj} G^j  = \sum\limits_{\mu=1}^{N}  \py_k (\mathrm{F} l^{(\mu)}_j ) G^{(\mu)j} \ \ \mathrm{or} \ \ G^i  = \sum\limits_{\mu=1}^{N}  g^{ik}\py_k (\mathrm{F} l^{(\mu)}_j ) G^{(\mu)j} \ .
\ea
Using (\ref{metric_relations}) (and the footnote \ref{footnote1}) we can split this into diagonal, ``free'', and the off-diagonal, ``interacting'', parts
\ba\label{rel_nonlin4}
G_i  = \sum\limits_{\mu=1}^{N}  G^{(\mu)}_{\ \ i} + \sum\limits_{\mu \ne \nu}^{N} \left(l^{(\nu)}_i l^{(\mu)}_j +\frac{\mathrm{F}^{(\nu)}}{\mathrm{F}^{(\mu)}} h^{(\mu)}_{ij}\right) G^{(\mu)j} \ .
\ea

Now, (\ref{rel_nonlin3}) or (\ref{rel_nonlin4}) can be used to calculate the relation between the non-linear connections as well as the other important geometrical quantities (various curvatures, torsion, etc). For example, a straightforward calculation leads to
\ba\label{Nonlinear_relation}
N^{i}_{\ k} = \sum\limits_{\mu=1}^{N} \left[G^{(\mu)j} \left( -C_k^{\ ir} + \frac{1}{2}g^{ir}\py_k\right)+ g^{ir}N^{(\mu)j}_{\ \ \ \ k} \right]\py_r (\mathrm{F} l^{(\mu)}_j ) \ ,
\ea
where we used $N^{i}_{\ j} = \frac{1}{2}\dot{\partial}_j G^{i}$ (see (\ref{ydN}) and the discussion above it) as well as the definition of the Cartan tensor (\ref{Cartan_tens}).

Because at the moment we do not have a good control over the structure of such objects in the general case, this is the subject of the ongoing research and the results will be reported elsewhere. In the next section, we will have more to say about this for the specific 2-dimensional case.

\section{Multimetric geometry: 2d example}\label{Multimetric 2d}

In this section, we would like to specify to the 2-dimensional case. The motivation for doing this is two-fold: firstly, it will serve as a nice (and more detailed) demonstration of some general results obtained in the previous section; secondly, for this case we will be able to explicitly study a very important object - a measure.

\subsection{Cartan equations}\label{Cartan_equations}

Because the main object defining Finsler geometry, the Finsler structure $\mathrm{F}$, is 1-homogeneous, it often makes it natural to work with the objects defined not on the tangent bundle, but rather on the projective sphere bundle $S\mathcal{M}$. In terms of the local coordinates, $S\mathcal{M}$ is defined by identifying a positive ray, $\{(x,\lambda x), \lambda>0\}$ with a single point of a projective sphere bundle. Then any object that is 0-homogeneous will be well defined on $S\mathcal{M}$. Applying this to the case of two dimensions proves to be particularly useful.

The advantage of being in two dimensions is, as it was briefly discussed at the end of section \ref{Preliminaries}, that the rank of $h_{ij}$ now equals 1 and this guarantees that the metric takes the form (\ref{2dmetric}) for some vector $m$. Then the basis on $S\mathcal{M}$ is given by
\ba\label{vectors}
e_1 = m^i \delta_i\ ,\ e_2 = l^i \delta_i\ ,\ e_3 = \mathrm{F}m^i \dot{\partial}_i\ ,
\ea
where $\delta_i = \partial_i - N^j_{\ i}\dot{\partial}_j$, as was defined earlier, and $l^i$ and $m^i$ are the Berwald zweinbeins defined in (\ref{2dmetric}-\ref{mi}). Note that $e_i$'s are really 0-homogeneous (this is the reason for the appearance of the factor of $\mathrm{F}$ in $e_3$).

In the same way, we introduce the dual projective sphere bundle, $S\mathcal{M}^*$, with the dual basis of 1-forms
\ba\label{1forms}
\omega^1 = m_i dx^i\ ,\ \omega^2 = l_i dx^i\ ,\ \omega^3 = \frac{1}{\mathrm{F}}m_i \delta y^i\ ,
\ea
where $\delta y^i = dy^i + N^i_{\ j}dx^j$ and, again, the factor of $\frac{1}{\mathrm{F}}$ is introduced to guarantee that $\omega^3$ is well defined on $S\mathcal{M}^*$. These 1-forms satisfy the usual Cartan equations on $S\mathcal{M}^*$ (\ref{Cartan_eqs}) and can be used as the definitions of the scalars $I$, $J$ and $K$, briefly discussed at the end of the section \ref{Preliminaries}. For the convenience we will write these equations one more time:
\ba\label{Cartan_eqs_1}
d\omega^1 = - I \omega^1 \wedge \omega^3 + \omega^2 \wedge \omega^3 \ ,\ d\omega^2 = - \omega^1 \wedge \omega^3\ ,\ d\omega^3 = K \omega^1 \wedge \omega^2 - J \omega^1 \wedge \omega^3 \ .
\ea

We start our analysis of the Cartan equations (\ref{Cartan_eqs_1}) for the 2-dimensional multimetric geometry (\ref{N_multimetric}) with the discussion of the invariant \textit{I}. In doing so, we will obtain some other results that will be useful in the following sections, in particular, in the study of the measure.

While the relations between $l-$ and $l^{(\mu)}-$zweinbeins are the same as in the general case (\ref{relations})
\ba
l_i = \dot{\partial}_i \mathrm{F} = \dot{\partial}_i \sum\limits_{\mu=1}^{N} \mathrm{F}^{(\mu)}\equiv \sum\limits_{\mu=1}^{N} l^{(\mu)}_i \ ,\label{l_i}\\
l^i = \frac{y^i}{\mathrm{F}} = \frac{\mathrm{F}^{(\mu)}}{\mathrm{F}}\frac{y^i}{\mathrm{F}^{(\mu)}} \equiv \frac{\mathrm{F}^{(\mu)}}{\mathrm{F}} l^{(\mu)i} \label{l^i}\ ,
\ea
for the $m$-zweinbein, we easily get, cf. (\ref{mi}) and (\ref{epsilon}):
\ba
m_i &=& \sqrt{{g}} \varepsilon_{ij}l^j = \frac{\sqrt{{g}}}{\sqrt{\alpha^{({\mu})}}} \epsilon^{({\mu})}_{ij}\frac{\mathrm{F}^{(\mu)}}{\mathrm{F}} l^{({\mu})j} \equiv \frac{\mathrm{F}^{(\mu)}}{\mathrm{F}} \sqrt\frac{{g}}{\alpha^{({\mu})}} m^{({\mu})}_i \ ,\label{m_i}\\
m^i &=& \frac{1}{\sqrt{{g}}} \varepsilon^{ij}l_j = \sum\limits_{{\mu}=1}^{N}\sqrt\frac{\alpha^{({\mu})}}{{g}} \epsilon^{({\mu})ij} l^{({\mu})}_j \equiv \sum\limits_{{\mu}=1}^{N}\sqrt\frac{\alpha^{({\mu})}}{{g}} m^{({\mu})i} \label{m^i}\ ,
\ea
where $\epsilon^{({\mu})}_{ij}:=\sqrt{\alpha^{({\mu})}}\varepsilon_{ij}$ and analogously for $\epsilon^{({\mu})ij}$ and $\alpha^{({\mu})}:=\det(\alpha^{({\mu})}_{ij})$. Using (\ref{l_i}) and (\ref{m_i}) in (\ref{lmrelation}) we get
\ba
m_i m_j = \mathrm{F}\dot{\partial}_i l_j = \mathrm{F}\dot{\partial}_i \sum\limits_{{\mu}=1}^{N} l^{({\mu})}_j =\sum\limits_{{\mu}=1}^{N} \frac{\mathrm{F}}{\mathrm{F}^{(\mu)}} m^{({\mu})}_i m^{({\mu})}_j = m_i m_j\sum\limits_{{\mu}=1}^{N}\left(\frac{\mathrm{F}}{\mathrm{F}^{(\mu)}}\right)^3 \frac{\alpha^{({\mu})}}{{g}}\ .
\ea
This leads to a useful relation between the determinants of the metrics that is valid only in 2d (at least, we were not able to find its analogue in the general case)
\ba\label{metric_relation}
\frac{{g}}{\mathrm{F}^3} = \sum\limits_{{\mu}=1}^{N}\frac{\alpha^{({\mu})}}{\mathrm{F}^{{(\mu)}3}}\ .
\ea

There are various equivalent expressions $I$-invariant {\cite{Bao,Antonelli}}. We will use the following one (which, of course, can be immediately derived from the Cartan equations (\ref{Cartan_eqs_1}))
\ba\label{Idef}
I = \frac{\mathrm{F}}{2{g}}m^i \dot{\partial}_i {g} \ .
\ea
From here it is not hard to see that using the definition of the Cartan tensor (\ref{Cartan_tens}), we have in 2-dimensional case
\ba\label{Cartan_tens_2d}
\mathrm{F}C_{ijk} = I m_i m_j m_k \ ,
\ea
i.e. we have a genuinely non-Riemannian geometry if and only if $I\ne 0$, as was commented above.

Using (\ref{metric_relation}), (\ref{m^i}) and recalling that $m^i l_i =0$ we immediately obtain
\ba\label{I}
I = -\frac{3}{2}\frac{\mathrm{F}^4}{{g}}\sum\limits_{{\mu}, {\nu}}^{N} \sqrt{\frac{\alpha^{({\mu})}}{{g}}}\frac{\alpha^{({\nu})}}{\mathrm{F}^{({\nu})4}} m^{({\mu})i}l^{(\nu)}_i \ .
\ea
We see that the ``non-Riemannianity'' is controlled by the ``cross-terms'' $m^{({\mu})i}l^{(\nu)}_i$. With the help of (\ref{l_i}) and (\ref{m^i}) these terms are easily expressed in terms of the corresponding metrics
\ba\label{Aml}
m^{({\mu})i}l^{({\nu})}_i = \frac{1}{\sqrt{\alpha^{({\mu})}}}\varepsilon^{ij}\frac{\alpha^{({\mu})}_{jr}y^r}{\mathrm{F}^{(\mu)}} \frac{\alpha^{({\nu})}_{in}y^n}{\mathrm{F}^{(\nu)}} =: -\sqrt{\frac{{g}}{\alpha^{({\mu})}}}\mathcal{A}^{({\mu}{\nu})} \ ,
\ea
where we introduced $\mathcal{A}^{({\mu}{\nu})} = \frac{\epsilon^{ji}\alpha^{({\mu})}_{jr}\alpha^{({\nu})}_{in}y^r y^n}{\mathrm{F}^{(\mu)} \mathrm{F}^{(\nu)}}$. With this notation, (\ref{I}) takes a very compact form
\ba\label{I1}
I = \frac{3}{2}\frac{\mathrm{F}^4}{{g}}\sum\limits_{{\mu},{\nu}}^{N} \frac{\alpha^{({\nu})}}{\mathrm{F}^{{(\nu)}4}} \mathcal{A}^{({\mu}{\nu})}\ .
\ea
$\mathcal{A}^{({\mu}{\nu})}$ satisfies a very nice relation, which will be useful later:
\ba\label{A_rel}
\frac{1}{\mathrm{F}^{(\nu)}} \sum\limits_{{\mu}=1}^N \mathcal{A}^{({\mu}{\nu})} m_j \stackrel{(\ref{m^i})}{=} -\frac{1}{\mathrm{F}^{(\nu)}} l^{(\nu)}_i m^i m_j \stackrel{(\ref{lmrelation})}{=} -\frac{\mathrm{F}}{\mathrm{F}^{(\nu)}} \py_j \left( \frac{\mathrm{F}^{(\nu)}}{\mathrm{F}}l^{(\nu)i}\right)l^{(\nu)}_i \stackrel{(\ref{lm})}{\equiv} \py_j \ln\frac{\mathrm{F}}{\mathrm{F}^{(\nu)}} \ .
\ea

It is trivial to see that $\mathcal{A}^{({\mu}{\nu})} = 0$ if and only if $\alpha^{({\mu})}_{ij} = \phi^{(\mu\nu)} \alpha^{({\nu})}_{ij}$ for some $\phi^{(\mu\nu)}(x)$, in full agreement with the general result found in Proposition \ref{PropB}. So, as we commented above, $\mathcal{A}^{({\mu}{\nu})}$ measures ``non-Riemannianity'' in $({\mu}{\nu})$ sector.

For the future analysis of the Cartan equations (\ref{Cartan_eqs_1}), it will be useful to specify the general analysis of the spray coefficient (\ref{rel_nonlin3}) and the non-linear connection (\ref{Nonlinear_relation}) to the 2-dimensional case. With the help of (\ref{l_i}-\ref{m^i}) this could be re-written as
\ba\label{rel_nonlin_2d}
G^i = \sum\limits_{\mu=1}^{N} \left[ l^i l^{(\mu)}_j + \left(\frac{\mathrm{F}}{\mathrm{F}^{(\mu)}}\right)^3 \frac{\alpha^{(\mu)}}{g} m^i m_j \right] G^{(\mu)j} = : \sum\limits_{\mu=1}^{N} \delta^{(\mu) i}_{\ \ \ \ j} G^{(\mu) j}\ ,
\ea
where we defined $\delta^{(\mu) i}_{\ \ \ \ j}$ satisfying the relation
\ba\label{delta}
\sum\limits_{\mu=1}^{N} \delta^{(\mu) i}_{\ \ \ \ j} = \delta^{i}_{\ j}\ .
\ea
Then either taking derivative of (\ref{rel_nonlin_2d}) with respect to $y$ or directly specifying (\ref{Nonlinear_relation}) to the 2-dimensional case, we obtain the explicit relation between the Finslerian and Riemannian (in each sector) non-linear connections
\ba\label{NN}
N^i_{\ j} &=& \frac{1}{2} \sum\limits_{\mu =1}^N \left(\dot{\partial}_j\delta^{(\mu) i}_{\ \ \ \ k}\right) G^{(\mu) k} + \sum\limits_{\mu=1}^{N} \delta^{(\mu) i}_{\ \ \ \ k} N^{(\mu) k}_{\ \ \ \ j} \equiv \nonumber \\
&\equiv & m^i m_j \sum\limits_{\mu=1}^{N} \left[ \frac{1}{2\mathrm{F}} l^{(\mu)}_k + \left(\frac{\mathrm{F}}{\mathrm{F}^{(\mu)}}\right)^3 \frac{\alpha^{(\mu)}}{g} \left( \frac{3}{2\mathrm{F}^{(\mu)}}\left( \sum\limits_{\nu=1}^{N}\mathcal{A}^{(\nu\mu)} \right)m_k - \frac{I}{\mathrm{F}}m_k - \frac{1}{2\mathrm{F}} l_k \right) \right] G^{(\mu)k} + \nonumber \\
&+& \sum\limits_{\mu=1}^{N} \left[ l^i l^{(\mu)}_k + \left(\frac{\mathrm{F}}{\mathrm{F}^{(\mu)}}\right)^3 \frac{\alpha^{(\mu)}}{g} m^i m_k \right] N^{(\mu)k}_{\ \ \ \ \ j}\ ,
\ea
where we also used (\ref{lmrelation}) and (\ref{Cartan_tens}).

We will also need the ``matrix'' elements of the non-linear connection, $N^i_{\ j}$, with respect to the $(l,m)$ basis. Contracting (\ref{NN}) with $l_i$, we, of course, get back (\ref{rel_nonlin1})
\ba\label{rel_nonlin11}
l_i N^i_{\ j} = \sum\limits_{\mu=1}^{N} l^{(\mu)}_i N^{(\mu)i}_{\ \ \ \,\ j} \ .
\ea
Introducing $\Delta N^{(\mu) i}_{\ \ \ \ j} := N^i_{\ j} - N^{(\mu) i}_{\ \ \ \ j}$, this can be written as
\ba\label{rel_nonlin12}
\sum\limits_{\mu=1}^{N} l^{(\mu)}_i \Delta N^{(\mu)i}_{\ \ \ \,\ j} =0 \ .
\ea
The contraction of (\ref{NN}) with $m_i$ has the less trivial form
\ba\label{mNlm}
m_i N^i_{\ j} l^j &=& \sum\limits_{\mu=1}^{N} \left(\frac{\mathrm{F}}{\mathrm{F}^{(\mu)}}\right)^3 \frac{\alpha^{(\mu)}}{g} m_i N^{(\mu)i}_{\ \ \ \,\ j} l^j \ ,\\
m_i N^i_{\ j} m^j &=& \sum\limits_{\mu=1}^{N} \left[ \frac{1}{2\mathrm{F}} l^{(\mu)}_k + \left(\frac{\mathrm{F}}{\mathrm{F}^{(\mu)}}\right)^3 \frac{\alpha^{(\mu)}}{g} \left( \frac{3}{2\mathrm{F}^{(\mu)}}\left( \sum\limits_{\nu=1}^{N}\mathcal{A}^{(\nu\mu)} \right)m_k - \frac{I}{\mathrm{F}}m_k - \frac{1}{2\mathrm{F}} l_k \right) \right] G^{(\mu)k} + \nonumber \\
&+& \sum\limits_{\mu=1}^{N} \left(\frac{\mathrm{F}}{\mathrm{F}^{(\mu)}}\right)^3 \frac{\alpha^{(\mu)}}{g}  m_k  N^{(\mu)k}_{\ \ \ \ \ j} m^j\ ,
\ea
which could be rewritten with the help of (\ref{metric_relation}) as
\ba
&&\sum\limits_{\mu=1}^{N} \left(\frac{\mathrm{F}}{\mathrm{F}^{(\mu)}}\right)^3 \frac{\alpha^{(\mu)}}{g} m_i \Delta N^{(\mu)i}_{\ \ \ \,\ j} l^j = 0\ , \label{mNl1}\\
&&\sum\limits_{\mu=1}^{N} \left(\frac{\mathrm{F}}{\mathrm{F}^{(\mu)}}\right)^3 \frac{\alpha^{(\mu)}}{g}  m_k  \Delta N^{(\mu)k}_{\ \ \ \ \ j} m^j= \nonumber\\
&&=\sum\limits_{\mu=1}^{N} \left[ \frac{1}{2\mathrm{F}} l^{(\mu)}_k + \left(\frac{\mathrm{F}}{\mathrm{F}^{(\mu)}}\right)^3 \frac{\alpha^{(\mu)}}{g} \left( \frac{3}{2\mathrm{F}^{(\mu)}}\left( \sum\limits_{\nu=1}^{N}\mathcal{A}^{(\nu\mu)} \right)m_k - \frac{I}{\mathrm{F}}m_k - \frac{1}{2\mathrm{F}} l_k \right) \right] N^{(\mu)k}_{\ \ \ \ \ j} y^j \ . \label{mNm1}
\ea

With the help of (\ref{A_rel}) we can establish one more relation (which will prove useful for the analysis of the Cartan equations, see below) for the matrix elements of $N^i_{\ j}$
\ba\label{AN}
\sum\limits_{\mu , \nu}^{N} \frac{\mathrm{F}^3}{\mathrm{F}^{(\mu)4}}\frac{\alpha^{(\mu)}}{g}\mathcal{A}^{(\nu\mu)} m_i \Delta N^{(\mu)i}_{\ \ \ \,\ j} l^j = \sum\limits_{\mu=1}^{N} \left(\frac{\mathrm{F}}{\mathrm{F}^{(\mu)}}\right)^3 \frac{\alpha^{(\mu)}}{g} \left( \frac{l_i}{\mathrm{F}} - \frac{l^{(\mu)}_i}{\mathrm{F}^{(\mu)}} \right) \Delta N^{(\mu)i}_{\ \ \ \,\ j} l^j \ .
\ea

Now, we are ready to study in details the Cartan equations (\ref{Cartan_eqs_1}). In this way, we will re-derive some of the general results as well as get a better control over the 2-dimensional case. We start with the first of the Cartan equations, $d\omega^1 = - I \omega^1 \wedge \omega^3 + \omega^2 \wedge \omega^3$. To do so, we need to establish the relations between the Finsler 1-forms, $\omega^i$ (\ref{1forms}) and their analogues in each Riemannian sector, $\omega^{(\mu)i}$. These can be easily found using the analogous relations between the Berwald zweinbeins (\ref{l_i}-\ref{m^i}), the result being
\ba
\omega^1 &=& \frac{\mathrm{F}^{(\mu)}}{\mathrm{F}}\sqrt{\frac{g}{\alpha^{(\mu)}}}\omega^{(\mu)1} = \sum\limits_{\mu=1}^{N}\left(\frac{\mathrm{F}}{\mathrm{F}^{(\mu)}}\right)^2\sqrt{\frac{\alpha^{(\mu)}}{g}}\omega^{(\mu)1}\ ,\label{1forms_relation1}\\
\omega^2 &=& \frac{\mathrm{F}}{\mathrm{F}^{(\mu)}}\omega^{(\mu)2} + \left(\sum\limits_{\nu=1}^{N} \mathcal{A}^{(\nu\mu)}\right)\sqrt{\frac{g}{\alpha^{(\mu)}}}\omega^{(\mu)1} = \sum\limits_{\mu=1}^{N}\omega^{(\mu)2}\ ,\label{1forms_relation2}\\
\omega^3 &=& \left(\frac{\mathrm{F}^{(\mu)}}{\mathrm{F}}\right)^2\sqrt{\frac{g}{\alpha^{(\mu)}}}\omega^{(\mu)3} + \frac{\mathrm{F}^{(\mu)}}{\mathrm{F}^2}\sqrt{\frac{g}{\alpha^{(\mu)}}}m^{(\mu)}_i \Delta N^{(\mu) i}_{\ \ \ \ j} dx^j = \nonumber \\
&=& \sum\limits_{\mu=1}^{N}\frac{\mathrm{F}}{\mathrm{F}^{(\mu)}}\sqrt{\frac{\alpha^{(\mu)}}{g}}\left[\omega^{(\mu)3} + \frac{1}{\mathrm{F}^{(\mu)}}m^{(\mu)}_i \Delta N^{(\mu) i}_{\ \ \ \ j} dx^j \right] \ ,\label{1forms_relation3}
\ea
where $dx^i = m^i \omega^1 + l^i \omega^2$. (To get the relation between $\omega^i$ and $\omega^{(\mu)i}$ for a fixed sector $\mu$, i.e. the first equality in each relation, we trivially used the fact that due to the bi-dimensionality, any of the pairs, $(l,m)$ or $(l^{(\mu)},m^{(\mu)})$, can be used as a basis.) These relations can be easily inverted
\ba
\omega^{(\mu)1} &=& \frac{\mathrm{F}}{\mathrm{F}^{(\mu)}}\sqrt{\frac{\alpha^{(\mu)}}{g}}\omega^1\ ,\label{1forms_relation1inverse}\\
\omega^{(\mu)2} &=& \frac{\mathrm{F}^{(\mu)}}{\mathrm{F}}\omega^2 - \left(\sum\limits_{\nu=1}^{N} \mathcal{A}^{(\nu\mu)}\right)\omega^{1}\ ,\label{1forms_relation2inverse}\\
\omega^{(\mu)3} &=& \left(\frac{\mathrm{F}}{\mathrm{F}^{(\mu)}}\right)^2\sqrt{\frac{\alpha^{(\mu)}}{g}}\left[\omega^3 - \frac{1}{\mathrm{F}}m_i \Delta N^{(\mu) i}_{\ \ \ \ j} dx^j \right]\ .\label{1forms_relation3inverse}
\ea
Using (\ref{1forms_relation1}-\ref{1forms_relation3inverse}), after some straightforward algebra, the first Cartan equation takes the form
\ba
d\omega^1 = A \omega^1 \wedge \omega^3 + B \omega^2 \wedge \omega^3 + C \omega^1 \wedge \omega^2 \ ,
\ea
where
\ba
A&=& -\frac{3}{2}\frac{\mathrm{F}^4}{{g}}\sum\limits_{\mu, \nu}^{N} \frac{\alpha^{(\mu)}}{\mathrm{F}^{(\mu)4}} \mathcal{A}^{(\nu\mu)}\stackrel{(\ref{I1})}{\equiv} -I \ ,\label{1-3} \\
B&=& \sum\limits_{\mu=1}^{N}\left(\frac{\mathrm{F}}{\mathrm{F}^{(\mu)}}\right)^3 {\frac{\alpha^{(\mu)}}{g}}\stackrel{(\ref{metric_relation})}{\equiv} 1 \ , \label{2-3} \\
C&=& -\frac{1}{2}\sum\limits_{\mu=1}^{N} \frac{\mathrm{F}^2}{\mathrm{F}^{(\mu)3}} {\frac{\alpha^{(\mu)}}{g}} l_i\Delta N^{(\mu) i}_{\ \ \ \ j}l^j + \nonumber\\
&& +\frac{3}{2}\sum\limits_{\mu, \nu}^{N} \frac{\mathrm{F}^3}{\mathrm{F}^{(\mu)4}} {\frac{\alpha^{(\mu)}}{g}} \mathcal{A}^{(\nu\mu)} m_i\Delta N^{(\mu) i}_{\ \ \ \ j}l^j + \sum\limits_{\mu=1}^{N} \frac{\mathrm{F}^2}{\mathrm{F}^{(\mu)3}} {\frac{\alpha^{(\mu)}}{g}} m_i\Delta N^{(\mu) i}_{\ \ \ \ j}m^j\ . \label{1-2}
\ea
We see that while $A$ agrees with the result for $I$ obtained above (\ref{I1}), the coefficient $B$ provides another way to derive the determinant relation (\ref{metric_relation}). It is less trivial to see that the last coefficient is actually zero, as it should be in general. To show this, we plug into (\ref{1-2}) the relations for the matrix elements that we found above, (\ref{mNm1}) and (\ref{AN}), and after some straightforward algebra, we see that (\ref{1-2}) is indeed identically zero.

Proceeding in the same way as above we can analyze the second Cartan equation, $d\omega^2 =  - \omega^1 \wedge \omega^3$. Then we have from (\ref{1forms_relation1}-\ref{1forms_relation3inverse})
\ba\label{d2}
d\omega^2 &=& A \omega^1 \wedge \omega^3 + B \omega^2 \wedge \omega^3 + C \omega^1 \wedge \omega^2 \ ,
\ea
where
\ba
A &=&  -\sum\limits_{\mu=1}^{N}\left(\frac{\mathrm{F}}{\mathrm{F}^{(\mu)}}\right)^3 {\frac{\alpha^{(\mu)}}{g}}\stackrel{(\ref{metric_relation})}{\equiv} -1\ ,\label{1-3next} \\
B &\equiv& 0\ ,\\
C &=& \sum\limits_{\mu=1}^{N} \frac{\mathrm{F}^2}{\mathrm{F}^{(\mu)3}} {\frac{\alpha^{(\mu)}}{g}} m_i\Delta N^{(\mu) i}_{\ \ \ \ j}l^j \ . \label{1-2next}
\ea
As before, (\ref{1-3next}) gives the trivially satisfied expected result (again, due to the determinant relation (\ref{metric_relation})) and (\ref{1-2next}) is nothing but the relation found above (\ref{mNl1}).

The third Cartan equation, $d\omega^3 = K \omega^1 \wedge \omega^2 - J \omega^1 \wedge \omega^3$ is the least trivial one. Proceeding as before and using the results for the first two Cartan equations we get (we omit the somewhat tedious but otherwise trivial manipulations)
\ba\label{d3}
d\omega^2 &=& A \omega^1 \wedge \omega^3 + B \omega^2 \wedge \omega^3 + C \omega^1 \wedge \omega^2 \ ,
\ea
where
\ba
A &=& \sum\limits_{\mu=1}^{N}\frac{\mathrm{F}^{3}}{\mathrm{F}^{(\mu)3}}\frac{\alpha^{(\mu)}}{g}\left[\left( \frac{3}{2} \frac{l^{(\mu)}_i}{\mathrm{F}^{(\mu)}} - \frac{l_i}{\mathrm{F}} \right)\Delta N^{(\mu) i}_{\ \ \ \ j} m^j -  m_i m^j m^r \py_r \left(\Delta N^{(\mu) i}_{\ \ \ \ j}\right) \right]\ ,\label{1-3nextnext} \\
B &=& \sum\limits_{\mu=1}^{N}\left[\frac{\mathrm{F}^{3}}{\mathrm{F}^{(\mu)4}}\frac{\alpha^{(\mu)}}{g}\left(-\frac{1}{2} \frac{\mathrm{F}^{(\mu)}}{\mathrm{F}}l_i + \frac{3}{2}\left(\sum\limits_{\nu=1}^{N} \mathcal{A}^{(\nu\mu)}\right) m_i\right)\Delta N^{(\mu) i}_{\ \ \ \ j} l^j + \right. \nonumber\\
&& \left. +\frac{\mathrm{F}^2}{\mathrm{F}^{(\mu)3}}{\frac{\alpha^{(\mu)}}{g}}m_i \Delta N^{(\mu) i}_{\ \ \ \ j} m^j \right] , \label{2-3nextnext} \\
C &=& \sum\limits_{\mu=1}^{N}\left[ K^{(\mu)}\frac{\mathrm{F}}{\mathrm{F}^{(\mu)}} {\frac{\alpha^{(\mu)}}{g}} - \frac{\mathrm{F}}{\mathrm{F}^{(\mu)}}\sqrt{\frac{\alpha^{(\mu)}}{g}}e_2 \left( \frac{\mathrm{F}}{\mathrm{F}^{(\mu)2}}\sqrt{\frac{\alpha^{(\mu)}}{g}}m_i \Delta N^{(\mu) i}_{\ \ \ \ j} m^j \right) -\right. \nonumber\\
&& \left. - \frac{\mathrm{F}}{\mathrm{F}^{(\mu)2}}\sqrt{\frac{\alpha^{(\mu)}}{g}}m_i \Delta N^{(\mu) i}_{\ \ \ \ j} l^j e_1 \left( \frac{\mathrm{F}}{\mathrm{F}^{(\mu)}}\sqrt{\frac{\alpha^{(\mu)}}{g}} \right)\right]\ . \label{1-2nextnext}
\ea
While $B=0$ due to the relation between the non-linear connections (it is actually the same as (\ref{1-2})), we should have $A = -J$ and $C =K$. Let us briefly discuss this result, postponing a more detailed analysis for the future.

Taking the exterior derivative of the first Cartan equation and using the rest of them, one can immediately see that the scalar $J$ can be expressed in terms of $I$ as follows
\ba\label{J1}
J = e_2(I)\ ,
\ea
where $e_2$ is given in (\ref{vectors}). This result is general and should be valid for any 2-dimensional Finsler geometry. It is indeed possible to verify that the coefficient $A$ given in (\ref{1-3nextnext}) satisfies this, i.e. $A = - e_2(I)$. The calculation is not very illuminating and we omit it. Still, this provides a very strong consistency check.

The importance of the scalar $J$ is that in the same way as the scalar $I$ controls the Riemannianity of a space, the scalar $J$ tells us whether it is of a Landsberg type or not. This can be easily seen from the fact that in 2d the Cartan tensor has only one non-trivial component (\ref{Cartan_tens_2d}). Then, from the Landsbergian condition (\ref{Lan}) we get that the condition takes a form: $\dot{C}_{111}\equiv J=0$, see e.g. \cite{Bao}. So, we can ask the same question that we addressed in the general case: what is the condition on our space to be of the Landsberg type? In the 2-dimensional bi-metric case the answer is very concrete:
\begin{proposition}\label{Prop3}
The 2-dimensional Finsler space defined by the Finsler structure (\ref{bimetric}) is of a Landsberg type, $J=0$, if and only if $I=0$, i.e. if the space is Riemannian.
\end{proposition}
\begin{proof}
The proof is essentially due to Szab$\acute{\mathrm{o}}$'s rigidity theorem \cite{Szabo} that states that there is no non-trivial Berwald space in 2d. Combining this with the fact that in the general n-dimensional case of a generalized $4^{th}$-root Finsler structure, Landsberg space is always of the Berwald type (\!\cite{Tabata}, Theorem 1.2), implies that in two dimensions $J=0$ if and only if $I=0$.\footnote{Though we have not checked it explicitly, we believe that the proof in \cite{Tabata} should be also true in the generalized case, i.e. for the Finsler structure given by (\ref{N_multimetric1}). Then the proposition should be valid for a general 2-dimensional multimetric Finsler geometry.}
\end{proof}

This result tells us that as soon as a 2d  bi-metric Finsler geometry is non-Riemannian, it will be of a general type. E.g., we will not have a non-trivial bi-metric space, for which the generalized Gauss-Bonnet theorem would still hold, because for that one needs $J=0 $\cite{Bao}.

Finally, we briefly discuss the last coefficient, $C$, (\ref{1-2nextnext}) in the third Cartan equation. As we mentioned, this coefficient should be equal to some scalar curvature. This is actually a scalar curvature of the Chern connection \cite{Bao}. This can be seen from the explicit evaluation of the relation $K = - \omega^3 [e_1 , e_2]$, which is an immediate consequence of the Cartan equations (\ref{Cartan_eqs_1}). Though one can observe some factorized structure in (\ref{1-2nextnext}), we do not have a complete control over this result, meaning that we cannot yet present it explicitly in the form: factorization of the sectors plus "cross-terms" as in the case of a metric (\ref{metric_relations}) or spray coefficients (\ref{rel_nonlin4}). We are planning to report on this elsewhere.

\subsection{Measure}

In this subsection, we would like to study a very important geometrical object - a measure on a Finsler manifold. There are two common choices, Holmes-Thompson \cite{Holmes-Thompson} and Busemann-Hausdorff \cite{Shen2001} measures. We will treat in detail the general case for the former one, though at the end, we give the Busemann-Hausdorff measure for the bi-metric case, too. We start with some preliminary technical considerations.

Consider a general Finsler function $\mathrm{F}$ (not necessarily 2-dimensional) with the corresponding metric $g_{ij}=\frac{1}{2}\py_i \py_j \mathrm{F}^2$, $i,j=\overline{1,n}$. Let, as before, ${g}=\det(g_{ij})$. We want to analyze the following integral
\ba
I_f = \int\limits_{\mathrm{F}\leq 1} f({g}) \omega \ ,
\ea
where $f$ is some arbitrary nice function, $\omega = \frac{1}{n!}\epsilon_{i_1\ldots i_n}dy^{i_1}\wedge\cdots\wedge dy^{i_n}$ and the integration region, $\mathrm{F}\leq 1$, is for each fixed $x$.

\begin{proposition}\label{Prop1}
The volume integral $I_f$ can be reduced to a surface integral over a unit sphere in $\mathbb{R}^n$
$$
I_f = \int\limits_{\|y\|=1}\frac{f({g})}{\mathrm{F}^n}\eta\ ,
$$
where $\omega = d\eta$ and $\|\cdot\|$ denotes the usual norm in $\mathbb{R}^n$.
\end{proposition}
\begin{proof}
First we show that $d(f({g})\eta) = f({g})\omega$. (All the differentials are with respect to $y$ variable only.) Really, this trivially follows from $d{g}={g}g^{ij}(\py_k g_{ij})dy^k$:
\ba
df({g})\wedge\eta = f'({g})d{g}\wedge\eta = \frac{1}{n!}f'({g}){g}g^{ij}(\py_k g_{ij})\epsilon_{i_1\ldots i_n}y^{i_1} dy^k \wedge dy^{i_2}\wedge\cdots\wedge dy^{i_n}\nonumber\\
= \frac{1}{n!}f'({g}){g}g^{ij}(y^{i_1}\py_k g_{ij}) \epsilon_{i_1\ldots i_n}\epsilon^{k i_2\ldots i_n}\omega = \frac{1}{n}f'({g}){g}g^{ij}(y^k\py_k g_{ij})\omega \equiv 0\ ,\nonumber
\ea
where in the last equality we used that $g_{ij}$ is 0-homogeneous with respect to $y$. Using this result, we have
$$
d(f({g})\eta)=df({g})\wedge\eta +f({g})\omega = f({g})\omega\ .
$$
Now, by Stokes' theorem we have
$$
\int\limits_{\mathrm{F}\leq 1} f({g}) \omega = \int\limits_{\mathrm{F}= 1} f({g}) \eta \equiv \int\limits_{\mathrm{F}= 1} \frac{f({g})}{\mathrm{F}^n} \eta\ ,
$$
where we trivially included the factor of $1/\mathrm{F}^n$, which equals 1 on the indicatrix, $\mathrm{F}=1$.

The next step is to show that this new form under the integral is actually exact, $d\left(\frac{f({g})}{\mathrm{F}^n} \eta\right)=0$. Really, using the above, we can write:
\ba
d\left(\frac{f({g})}{\mathrm{F}^n} \eta\right) =f({g})d\left(\frac{1}{\mathrm{F}^n} \eta\right)=f({g})\left(d\left(\frac{1}{\mathrm{F}^n}\right)\wedge \eta + \frac{1}{\mathrm{F}^n}\omega\right) \nonumber \\ =f({g})\left(-\frac{n}{\mathrm{F}^{n+1}}\py_k F dy^k\wedge \eta + \frac{1}{\mathrm{F}^n}\omega\right) = f({g})\left(-\frac{1}{\mathrm{F}^{n+1}}(y^k\py_k \mathrm{F})\omega + \frac{1}{\mathrm{F}^n}\omega\right)\equiv 0\ , \nonumber
\ea
where we used exactly the same manipulations as above and the fact that $y^k\py_k \mathrm{F} = \mathrm{F}$ by 1-homogeneity of $\mathrm{F}$.

The last step is to note that the above result means that the integrals of the form $\frac{f(g)}{\mathrm{F}^n} \eta$ over any two surfaces around the origin, $y=0$, (with a trivial winding number) will be equal. Take as the first surface the indicatrix, $\mathrm{F}=1$, and as the second one a unit sphere, $\|y\|=1$. Then we have our proof
$$
\int\limits_{\mathrm{F}\leq 1} f({g}) \omega = \int\limits_{\mathrm{F}= 1} \frac{f({g})}{\mathrm{F}^n} \eta = \int\limits_{\|y\|= 1} \frac{f({g})}{\mathrm{F}^n} \eta\ . \nonumber
$$
\end{proof}

Now, we apply the result of Proposition \ref{Prop1} to determine the Holmes-Thompson (HT) volume form for the case of the 2-dimensional multimetric Finsler space with the Finsler function (\ref{N_multimetric}).

The Holmes-Thompson  volume form is based on the preferred Riemannian metric on $T\mathcal{M}\!\setminus \!\{0\}$ associated with the Finsler structure - the Sassaki metric \cite{Chern} (we will not need the explicit form of this metric). This volume form is given by
\ba\label{HT}
\mu_{HT}(x)= \frac{\int\limits_{\mathrm{F}\leq 1} {g}\omega}{Vol_{\mathbb{R}^n}(\|y\|\leq 1)} \ ,
\ea
where the denominator is just the usual volume of a unit ball, $Vol_{\mathbb{R}^n}(\|y\|\leq 1) = \frac{\pi^{\frac{n}{2}}}{\Gamma\left(\frac{n}{2} + 1\right)}$.
Using the result of Proposition \ref{Prop1}, we know that we can reduce the integral in (\ref{HT}) to an integral over a unit circle:
\ba
\int\limits_{\mathrm{F}\leq 1} {g}\omega = \int\limits_{\|y\| = 1} \frac{{g}}{\mathrm{F}^2}\eta \ .
\ea
Then we have for the Holmes-Thompson measure the following result
\begin{proposition}\label{Holmes-Thomson}
The Holmes-Thompson measure for the 2-dimensional Finsler space defined by (\ref{N_multimetric}) is given by
\ba\label{HTfinal}
\mu_{HT}(x)= \sum\limits_{\mu=1}^{N} \sqrt{\alpha^{(\mu)}} + \frac{2}{\pi}\sum\limits_{\mu\ne\nu}^{N} \sqrt{{\alpha^{(\mu)}}{\lambda^{(\mu\nu)}_+}}\,E\left(\sqrt{1-\frac{\lambda^{(\mu\nu)}_-}{\lambda^{(\mu\nu)}_+}}\right)\ ,
\ea
where $E(k)$ is the complete elliptic integral of the second kind (of some argument $k$) and $\lambda^{(\mu\nu)}_\pm$ are given by
\ba
\lambda^{(\mu\nu)}_{\pm} = \frac{e_1(\alpha^{(\mu)-1}\alpha^{(\nu)})\pm \sqrt{(e_1(\alpha^{(\mu)-1}\alpha^{(\nu)}))^2-4e_2(\alpha^{(\mu)-1}\alpha^{(\nu)})}}{2e_2(\alpha^{(\mu)-1}\alpha^{(\nu)})} \nonumber
\ea
with $e_1(\mathds{X})=\Tr \mathds{X}$ and $e_2(\mathds{X})=\frac{1}{2}\left((\Tr \mathds{X})^2 - \Tr \mathds{X}^2\right)$ being the symmetric polynomials associated with the matrix $\mathds{X}$.
\end{proposition}
\begin{proof}
We will use the relation between ${g}$, $\alpha^{(\mu)}$ found above (\ref{metric_relation})
\ba
\frac{{g}}{\mathrm{F}^3} = \sum\limits_{{\mu}=1}^{N}\frac{\alpha^{({\mu})}}{\mathrm{F}^{{(\mu)}3}} \ .
\ea
From here we immediately have (recall that $\alpha^{(\mu)}$ does not depend on $y$)
\ba\label{HT1}
\int\limits_{\|y\| = 1} \frac{{g}}{\mathrm{F}^2}\eta = \sum\limits_{{\mu}=1}^{N} \alpha^{(\mu)}\!\!\! \int\limits_{\|y\| = 1} \frac{\mathrm{F}}{\mathrm{F}^{(\mu)3}}\eta
=\sum\limits_{{\mu}=1}^{N} \alpha^{(\mu)}\!\!\! \int\limits_{\|y\| = 1} \frac{1}{\mathrm{F}^{(\mu)2}}\eta + \sum\limits_{\mu\ne\nu}^{N} \alpha^{(\mu)}\!\!\! \int\limits_{\|y\| = 1} \frac{\mathrm{F}^{(\nu)}}{\mathrm{F}^{(\mu)3}}\eta  \ .
\ea
To calculate the second integral in (\ref{HT1}) (the first one we will get by setting $\nu = \mu$ in the result), we note that this is an integral over a unit circle. Then, after the change of variables $t=\cot \phi$, where $\phi$ is the polar angle,\footnote{I.e., $\mathrm{F}^{(\mu)2}|_{\|y\|=1}=\alpha^{(\mu)}_{11} \cos(\phi)^2 + 2\alpha^{(\mu)}_{12}\cos(\phi)\sin(\phi)+\alpha^{(\mu)}_{22}\sin(\phi)^2$.} we get
\ba\label{HTint}
\int\limits_{\|y\| = 1} \frac{\mathrm{F}^{(\nu)}}{\mathrm{F}^{(\mu)3}} \eta = \int\limits_{-\infty}^{\infty} dt\frac{\sqrt{\alpha^{(\nu)}_{11} t^2 + 2\alpha^{(\nu)}_{12}t+\alpha^{(\nu)}_{22}}}{(\alpha^{(\mu)}_{11} t^2 + 2\alpha^{(\mu)}_{12}t+\alpha^{(\mu)}_{22})^{3/2}} \ .
\ea
Integrals of this type are treated in Appendix \ref{appendix} with the result being
\ba\label{HTint1}
\int\limits_{\|y\| = 1} \frac{\mathrm{F}^{(\nu)}}{\mathrm{F}^{(\mu)3}} \eta = 2\lambda^{(\mu\nu)}_+\sqrt{\frac{\lambda^{(\mu\nu)}_-}{\alpha^{(\nu)}}}\,E\left(\sqrt{1-\frac{\lambda^{(\mu\nu)}_-}{\lambda^{(\mu\nu)}_+}}\right) \equiv 2\sqrt{\frac{\lambda^{(\mu\nu)}_+}{\alpha^{(\mu)}}}\,E\left(\sqrt{1-\frac{\lambda^{(\mu\nu)}_-}{\lambda^{(\mu\nu)}_+}}\right)\ ,
\ea
where $E(k)$ is the complete elliptic integral of the second kind (of the argument $k$) and $\lambda^{(\mu\nu)}_\pm$ are defined as in (\ref{lambda2}). For the detailed discussion and all the steps of the calculation, see Appendix, where we re-derive some known results in a more geometric form.

Setting $\mu = \nu$ in (\ref{HTint1}), we immediately get
\ba\label{FHB}
\int\limits_{\|y\| = 1} \frac{1}{\mathrm{F}^{(\mu)2}}\eta = \frac{\pi}{\sqrt{\alpha^{(\mu)}}} \ ,
\ea
where we used that $\lambda^{(\mu\mu)}_\pm = 1$ and $E(0)=\pi /2$. Plugging (\ref{HTint1}) and (\ref{FHB}) into (\ref{HT1}) (and dividing by $Vol_{\mathbb{R}^2}(\|y\|\leq 1)=\pi$) we get the result (\ref{HTfinal}) .
\end{proof}

Note the anticipated factorized structure of (\ref{HTfinal}): while the first term is just a sum of the individual Riemannian measures, the second one describes the ``interaction'' between the geometries. The appearance of the symmetric polynomials of the matrices $\left(\alpha^{(\mu)}\right)^{-1}\alpha^{(\nu)}$ in those ``interaction'' terms is very suggestive as exactly these are the building blocks of the interaction part in the bi-metric gravity \cite{Heisenberg} (also see (\ref{action_massive})). This requires further study.

The second choice for the volume form is called Busemann-Hausdorff measure and it is defined as
\ba\label{HB}
\mu_{BH}(x)=\frac{Vol_{\mathbb{R}^n}(\|y\|\leq 1)}{Vol_{\mathbb{R}^n}(\mathrm{F}\leq 1)}\ .
\ea
Because it seems to be much less relevant for possible physical applications and because the analytic result lacks the factorized structure of the HT measure (\ref{HTfinal}), we will briefly consider just a bi-metric case. So, consider a two-dimensional Finsler space with the Finsler function $\mathrm{F}=\mathrm{F}_{\alpha} + \mathrm{F}_{\beta}$, where $\mathrm{F}_{\alpha} = \sqrt{\alpha(x)_{ij}y^i y^j}$ and $\mathrm{F}_{\beta} = \sqrt{\beta(x)_{ij}y^i y^j}$ are the usual Riemannian structures, as usual. Then we have the following
\begin{proposition}
For a two-dimensional Finsler space with the Finsler function $\mathrm{F}=\mathrm{F}_{\alpha} + \mathrm{F}_{\beta}$ (notations as above), the Busemann-Hausdorff volume form is given by
$$
\mu_{BH}(x)= \frac{1}{\mathcal{A}+\mathcal{B}}\ ,
$$
where $\mathcal{A}=\frac{1}{2}\frac{\Tr h_+ h_-^{-1}}{\sqrt{\det h_-}}$ with $h_\pm := \alpha \pm \beta$ and $\mathcal{B}$ is some elliptic integral given below.
\end{proposition}
\begin{proof}
First, let us manipulate the expression for $\frac{1}{\mathrm{F}^2}$.
\ba
\frac{1}{\mathrm{F}^2} = \frac{1}{\left(\mathrm{F}_{\alpha} + \mathrm{F}_{\beta}\right)^2} = \frac{\left(\mathrm{F}_{\alpha} - \mathrm{F}_{\beta}\right)^2}{\left(\mathrm{F}_{\alpha} + \mathrm{F}_{\beta}\right)^2 \left(\mathrm{F}_{\alpha} - \mathrm{F}_{\beta}\right)^2} = \frac{\mathrm{F}_{\alpha}^2 + \mathrm{F}_{\beta}^2}{\left(\mathrm{F}_{\alpha}^2 - \mathrm{F}_{\beta}^2\right)^2} - 2\frac{\mathrm{F}_{\alpha} \mathrm{F}_{\beta}}{\left(\mathrm{F}_{\alpha}^2 - \mathrm{F}_{\beta}^2\right)^2}\ . \nonumber
\ea
Now, we define
\ba
&&\mathrm{F}_{\alpha}^2 \pm \mathrm{F}_{\beta}^2 = \alpha(x)_{ij}y^i y^j \pm \beta(x)_{ij}y^i y^j = (\alpha(x)_{ij} \pm \beta(x)_{ij})y^i y^j =: h_{\pm}(x)_{ij}y^i y^j \ . \nonumber
\ea
Using this definition, we can rewrite
\ba\label{FFF}
\frac{1}{\mathrm{F}^2} = \frac{\mathrm{F}_+^2}{\mathrm{F}_-^4} - 2\frac{\mathrm{F}_{\alpha} \mathrm{F}_{\beta}}{\mathrm{F}_-^4}\ ,
\ea
where, formally, $\mathrm{F}_{\pm}=\sqrt{h_{\pm}(x)_{ij}y^i y^j}$ are the Finsler functions corresponding to the Riemannian metrics $h_{\pm}(x)$. (Note that $\mathrm{F}_{-}$ always appears as some even degree, so there will never be a problem with a possible non-positivness of $h_{-}$.)

Next we note that we can write (suppressing the irrelevant variable $x$)
\ba
h_{+ij}\frac{\partial}{\partial h_{-ij}} \mathrm{F}_-^2 = h_{+ij}\frac{\partial}{\partial h_{-ij}} \left( h_{-kl}y^k y^l \right) = h_{+ij}y^i y^j \equiv \mathrm{F}_+^2 \ . \nonumber
\ea
Using this we have for the first term in (\ref{FFF})
\ba
\frac{\mathrm{F}_+^2}{\mathrm{F}_-^4} = - h_{+ij}\frac{\partial}{\partial h_{-ij}} \frac{1}{\mathrm{F}_-^2} \ .
\ea
The integral of $\frac{1}{\mathrm{F}_-^2}$ over a circle, $\|y\|=1$ immediately gives as in (\ref{FHB})
\ba\label{Riemann}
\int\limits_{\|y\|= 1} \frac{1}{\mathrm{F}_-^2} \eta = \frac{\pi}{\sqrt{\det h_-}}\ .
\ea
Then we have for the integral of the first term in (\ref{FFF})
\ba\label{A}
\int\limits_{\|y\|= 1} \frac{\mathrm{F}_+^2}{\mathrm{F}_-^4} \eta = - h_{+ij}\frac{\partial}{\partial h_{-ij}} \frac{\pi}{\sqrt{\det h_-}} = \frac{\pi}{2}\frac{\Tr h_+ h_-^{-1}}{\sqrt{\det h_-}}  \ ,
\ea
where we used $\delta \sqrt{\det h_-} = \frac{1}{2}\sqrt{\det h_-} \Tr (h_-^{-1}\delta h_-)$. The equation (\ref{A}) gives us the value of $\mathcal{A}$ (after using $Vol_{\mathbb{R}^2}(\|y\|\leq 1)=\pi$).

Now we will deal with the second term in (\ref{FFF}). After the same change of variables as in (\ref{HTint}), we get
\ba\label{B}
\int\limits_{\|y\|= 1} \frac{\mathrm{F}_{\alpha} \mathrm{F}_{\beta}}{\mathrm{F}_-^4} \eta = 2 \int\limits_{-\infty}^{\infty} \frac{\sqrt{(\alpha_{11}t^2 + 2\alpha_{12}t+\alpha_{22})(\beta_{11}t^2 + 2\beta_{12}t+\beta_{22})}}{(h_{-11}t^2 + 2h_{-12}t+h_{-22})^2}dt \ .
\ea
Equation (\ref{B}) is an integral related to the elliptic integral of the second kind. This proves the statement about $\mathcal{B}$ and we finally have
\ba\label{HB1}
\int\limits_{\|y\|= 1} \frac{1}{\mathrm{F}^2} \eta = \frac{\pi}{2}\frac{\Tr h_+ h_-^{-1}}{\sqrt{\det h_-}} - 4 \int\limits_{-\infty}^{\infty} \frac{\sqrt{(\alpha_{11} + 2\alpha_{12}t+\alpha_{22}t^2)(\beta_{11} + 2\beta_{12}t+\beta_{22}t^2)}}{(h_{-11} + 2h_{-12}t+h_{-22}t^2)^2}dt \ .
\ea
\end{proof}

\section{Conclusion}\label{Conclusion}

In this work, we have initiated the study of the multimetric Finsler geometry. Though our motivation comes from the bi-metric formulation of massive gravity \cite{deRham:2014zqa}, the main focus of the paper is on the study of general properties of the geometry. In particular, whenever it was possible, we tried to highlight the factorized structure of the important geometric objects (spray coefficient, non-linear connection, measure among others). For the general case, we have established some relations between the Finslerian geometric objects and the Riemannian ingredients. The emergent structure has a factorized form: the direct sum over the Riemannian sectors is supplemented by the terms describing the ``interaction'' between the structures in Riemannian geometries. We further study this for the special case of two dimensions, where we manage to find a closed form of the Holmes-Thompson measure. This measure also reveals the same factorization as well as leads to the natural appearance of the symmetric polynomials relevant for the bi-metric formulation of massive gravity.

In spite of the significant initial results, the further study is necessary to get a more complete understanding of the multimetric geometry before it would be possible to apply it to physical problems. One of the most urgent tasks is to study different curvature tensors as the building blocks in constructing Finslerian gravities. As we saw, this problem proves to be quite technically involved even in the simplest 2-dimensional case. The general $n$-dimensional case is much more complicated. We hope to address this problem in the near future.

Another important question to be addressed is the construction of the measure in the general case. Though in 2d we have succeeded in finding the measure relevant for possible physical applications, the Holmes-Thompson one, we do not expect that the same approach would work in the case of the arbitrary number of dimensions. This requires further study.

In the introduction we mentioned the spectral action approach that allows to construct generalizations of the Einstein gravity starting with some generalized geometries. It would be interesting to see if gravity based on the multimetric Finsler geometry could be derived from the spectral action principle. For this, one needs to study natural Dirac operators for this geometry \cite{Flaherty1996,Flaherty1998}. This would also allow to address some geometrical/metrical questions from the point of view of the spectral geometry \cite{Connes:1994yd}.

As another potential physical application of the multimetric Finsler geometry, we would like to mention its possible use as a natural model for quantum gravitational fluctuations. For this, one would need a further generalization: instead of the discrete family of Riemannian spaces, one should introduce a continuous family. Then the continuous generalization of (\ref{N_multimetric}) (for this we would need to introduce some measure on the space of metrics) would correspond to a fluctuating geometry with (\ref{action_Finsler}) providing a natural action for a point particle coupled to this fluctuating geometry.

\section*{Acknowledgement}
AP acknowledges the partial support of CNPq under the grant no.312842/2021-0. The research of CL is supported by the CNPq PhD fellowship and of PC and RM by the CAPES fellowships.

\appendix
\section{On the relevant elliptic integrals}\label{appendix}

Here we want to study in details the integral appearing in Proposition \ref{Holmes-Thomson} and given in (\ref{HTint}) (for convenience, we switch to the bi-metric notations, i.e. $\alpha$ and $\beta$ below are the usual Riemannian structures):
\ba\label{elliptic2}
\mathfrak{J}=\int\limits_{-\infty}^{\infty} \frac{\sqrt{(\alpha_{11}t^2 + 2\alpha_{12}t+\alpha_{22})}}{(\beta_{11}t^2 + 2\beta_{12}t+\beta_{22})^{3/2}}dt \ .
\ea
Though this integral can be treated as any usual elliptic integral, see, e.g. \cite{Bateman2006}, the standard approach hides the geometric nature of the answer. Here we want to bring this integral to the Legendre canonical form with all the parameters written in terms of some geometric invariants constructed from the Riemannian metrics $\alpha$ and $\beta$. To make the treatment more self contained and for the future use, we start with the closely related integral
\ba\label{elliptic3}
\mathfrak{I}=\int\limits_{-\infty}^{\infty} \frac{dt}{\sqrt{(\alpha_{11}t^2 + 2\alpha_{12}t+\alpha_{22})(\beta_{11}t^2 + 2\beta_{12}t+\beta_{22})}} \ .
\ea

\begin{proposition}\label{FFintegral}
The integral (\ref{elliptic3}) is given by
\ba\label{final3}
\mathfrak{I} = 2\sqrt{\frac{\lambda_-}{\det\alpha}}\,K\left(\sqrt{1-\frac{\lambda_-}{\lambda_+}}\right)\ ,
\ea
where $K(k)$ is the complete elliptic integral of the first kind and $\lambda_{\pm}$ given by
\ba
\lambda_{\pm} = \frac{e_1(\alpha^{-1}\beta)\pm \sqrt{(e_1(\alpha^{-1}\beta))^2-4e_2(\alpha^{-1}\beta)}}{2e_2(\alpha^{-1}\beta)}
\ea
with $e_1(\mathds{X})=\Tr \mathds{X}$ and $e_2(\mathds{X})=\frac{1}{2}\left((\Tr \mathds{X})^2 - \Tr \mathds{X}^2\right)$ being the symmetric polynomials associated with the matrix $\mathds{X}$.
\end{proposition}
\begin{proof}
Let us denote $\alpha (t) := \alpha_{11}t^2 + 2\alpha_{12}t+\alpha_{22}$ and $\beta (t) := \beta_{11}t^2 + 2\beta_{12}t+\beta_{22}$. We want to bring $\alpha(t)$ and $\beta(t)$ to the standard form. Though the procedure is well known, it is given in terms of the matrix elements $\alpha_{ij}$, $\beta_{ij}$ hiding the geometric structures. Our proof will keep those structures explicit. Though the construction does not look symmetric with respect to the exchange of $\alpha$ and $\beta$ we will show that this asymmetry is apparent.

As the first step, let us find such $\lambda$'s that the polynomial $\alpha(t) - \lambda\beta(t)$ is a perfect square. These $\lambda$'s will be given by the solutions of the characteristic equation $\det (\alpha - \lambda\beta)=0$:
\ba\label{lambda}
&&\det (\alpha - \lambda\beta) = \det\alpha (1 - \lambda\, e_1(\alpha^{-1}\beta) + \lambda^2 \, e_2(\alpha^{-1}\beta)) = 0 \ \ \Rightarrow\nonumber \\
&& \lambda_{\pm} = \frac{e_1(\alpha^{-1}\beta)\pm \sqrt{(e_1(\alpha^{-1}\beta))^2-4e_2(\alpha^{-1}\beta)}}{2e_2(\alpha^{-1}\beta)} \ .
\ea
Note that because $\det(\alpha - 0\beta)=\det\alpha > 0$ and $\det(\alpha - \frac{\alpha_{11}}{\beta_{11}}\beta)=-(\alpha_{12} - \frac{\alpha_{11}}{\beta_{11}}\beta_{12})^2 \leq 0$, we have that both $\lambda$'s are real. Also from the fact that $e_1(\alpha^{-1}\beta)$ and $e_2(\alpha^{-1}\beta)$ are positive, we conclude that $\lambda_+ \geq \lambda_- > 0$. For each $\lambda$ there is a zero eigen-vector, $v_{\pm}$, of the matrix $\alpha - \lambda_{\pm}\beta$:
\ba\label{gammas}
&&(\alpha - \lambda_{\pm}\beta)v_{\pm}=0\ \ \mathrm{or}\ \ \left(
                                                           \begin{array}{cc}
                                                             \alpha_{11} - \lambda_{\pm}\beta_{11} & \alpha_{12} - \lambda_{\pm}\beta_{12} \\
                                                             \alpha_{21} - \lambda_{\pm}\beta_{21} & \alpha_{22} - \lambda_{\pm}\beta_{22} \\
                                                           \end{array}
                                                         \right) \left(
                                                                   \begin{array}{c}
                                                                     \gamma_{\pm} \\
                                                                     1 \\
                                                                   \end{array}
                                                                 \right) = 0\ \ \Rightarrow\nonumber \\
&&\gamma_{\pm} = -\frac{\alpha_{12}-\lambda_{\pm}\beta_{12}}{\alpha_{11}-\lambda_{\pm}\beta_{11}}\equiv -\frac{\alpha_{22}-\lambda_{\pm}\beta_{22}}{\alpha_{12}-\lambda_{\pm}\beta_{12}}\ ,
\ea
where we normalized $v_{\pm}$ so that its second component equals 1. Introducing a vector $u=\left(t \atop 1\right)$, we can write
\ba\label{perfect}
\alpha(t) - \lambda_{\pm}\beta(t) = u^{\top} (\alpha - \lambda_{\pm}\beta) u \equiv (u^{\top} - v_{\pm}^{\top} ) (\alpha - \lambda_{\pm}\beta) (u - v_{\pm}) = (\alpha_{11} - \lambda_{\pm}\beta_{11})(t-\gamma_{\pm})^2
\ea
proving the announced representation as a perfect square. From (\ref{perfect}) we can easily find
\ba\label{AB}
&&\alpha(t) = \frac{1}{\lambda_+ - \lambda_-}\left( \lambda_+ (\alpha_{11} - \lambda_- \beta_{11})(t-\gamma_-)^2 - \lambda_- (\alpha_{11} - \lambda_+ \beta_{11})(t-\gamma_+)^2 \right)\ , \nonumber \\
&&\beta(t) = \frac{1}{\lambda_+ - \lambda_-}\left( (\alpha_{11} - \lambda_- \beta_{11})(t-\gamma_-)^2 - (\alpha_{11} - \lambda_+ \beta_{11})(t-\gamma_+)^2 \right)\ .
\ea
Because $\forall t \in \mathds{R}$ one has $\alpha(t),\,\beta(t)>0$, we conclude that
\ba\label{positive}
&&\alpha(\gamma_-) = -\frac{\lambda_-}{\lambda_+ - \lambda_-} (\alpha_{11} - \lambda_+ \beta_{11})(\gamma_+ -\gamma_-)^2\ \ \Rightarrow\ \ (\alpha_{11} - \lambda_+ \beta_{11}) <0 \ , \nonumber \\
&&\alpha(\gamma_+) = \frac{\lambda_+}{\lambda_+ - \lambda_-} (\alpha_{11} - \lambda_- \beta_{11})(\gamma_+ -\gamma_-)^2\ \ \Rightarrow\ \ (\alpha_{11} - \lambda_- \beta_{11}) >0 \ .
\ea
On the other hand we can calculate $\alpha(\gamma_\pm)$ explicitly. From (\ref{gammas}) follows that $\gamma_\pm^2 = \frac{\alpha_{22}-\lambda_{\pm}\beta_{22}}{\alpha_{11}-\lambda_{\pm}\beta_{11}}$ and then we have
\ba\label{nice}
\alpha(\gamma_\pm) &=& \alpha_{11}\frac{\alpha_{22}-\lambda_{\pm}\beta_{22}}{\alpha_{11}-\lambda_{\pm}\beta_{11}} -2\alpha_{12}\frac{\alpha_{12}-\lambda_{\pm}\beta_{12}}{\alpha_{11}-\lambda_{\pm}\beta_{11}} + \alpha_{22}= \nonumber \\ &=&\frac{2(\alpha_{11}\alpha_{22}-(\alpha_{12})^2) - \lambda_\pm (\alpha_{11}\beta_{22}-2\alpha_{12}\beta_{12}+\alpha_{22}\beta_{11})}{\alpha_{11}-\lambda_{\pm}\beta_{11}} \ .
\ea
Recalling that for $\alpha^{-1}$ one has
\ba
\alpha^{-1}=\frac{1}{\det\alpha}\left(
                                  \begin{array}{cc}
                                    \alpha_{22} & -\alpha_{12} \\
                                    -\alpha_{12} & \alpha_{11} \\
                                  \end{array}
                                \right) \ ,
\ea
we can rewrite (\ref{nice}) as
\ba\label{nice2}
\alpha(\gamma_\pm) = \det\alpha\frac{2-\lambda_{\pm}e_1(\alpha^{-1}\beta)}{\alpha_{11}-\lambda_{\pm}\beta_{11}}\ .
\ea
Comparing (\ref{positive}) and (\ref{nice2}) we get a relation that will prove useful later
\ba\label{nice3}
\pm \frac{\lambda_\pm}{\lambda_+ - \lambda_-}(\alpha_{11} - \lambda_+ \beta_{11})(\alpha_{11} - \lambda_- \beta_{11})(\gamma_+ - \gamma_-)^2 = \det\alpha (2-\lambda_{\pm}e_1(\alpha^{-1}\beta))\ .
\ea
Using $\lambda_+ + \lambda_- = \frac{e_1(\alpha^{-1}\beta)}{e_2(\alpha^{-1}\beta)}$ and $\lambda_+ \lambda_- = \frac{1}{e_2(\alpha^{-1}\beta)}$, one can further simplify (\ref{nice3}) (we keep only the sign used later)
\ba\label{nice4}
\frac{\lambda_+}{\lambda_+ - \lambda_-}(\alpha_{11} - \lambda_+ \beta_{11})(\alpha_{11} - \lambda_- \beta_{11})(\gamma_+ - \gamma_-)^2 = -\det\alpha \,\frac{\lambda_+ - \lambda_-}{\lambda_-}\ .
\ea
Now we finally proceed to bringing $\mathfrak{I}$ to the canonical form. Rewrite (\ref{AB}) as
\ba\label{AB1}
&&\alpha(t) = -\frac{\lambda_- (\alpha_{11} - \lambda_+ \beta_{11})}{\lambda_+ - \lambda_-}(t-\gamma_+)^2\left(1 + \left|\frac{\lambda_+ (\alpha_{11} - \lambda_- \beta_{11})}{\lambda_- (\alpha_{11} - \lambda_+ \beta_{11})}\right|\left(\frac{t-\gamma_-}{t-\gamma_+}\right)^2 \right)\ , \nonumber \\
&&\beta(t) = -\frac{(\alpha_{11} - \lambda_+ \beta_{11})}{\lambda_+ - \lambda_-}(t-\gamma_+)^2\left(1 + \left|\frac{\alpha_{11} - \lambda_- \beta_{11}}{\alpha_{11} - \lambda_+ \beta_{11}}\right|\left(\frac{t-\gamma_-}{t-\gamma_+}\right)^2 \right)\ .
\ea
(Recall that $\alpha_{11} - \lambda_+ \beta_{11}$ is negative.) Now introduce a new variable
\ba\label{xi}
\xi = \left|\frac{\lambda_+ (\alpha_{11} - \lambda_- \beta_{11})}{\lambda_- (\alpha_{11} - \lambda_+ \beta_{11})}\right|^{1/2} \frac{t-\gamma_-}{t-\gamma_+}\ ,
\ea
which defines a good change of variables for $\xi$ belonging to either of two intervals:
$$\left(-\mathrm{sign}(\gamma_+ - \gamma_-)\infty, \left|\frac{\lambda_+ (\alpha_{11} - \lambda_- \beta_{11})}{\lambda_- (\alpha_{11} - \lambda_+ \beta_{11})}\right|^{1/2}\right]\ \ \mathrm{or}\ \ \left[\left|\frac{\lambda_+ (\alpha_{11} - \lambda_- \beta_{11})}{\lambda_- (\alpha_{11} - \lambda_+ \beta_{11})}\right|^{1/2},\mathrm{sign}(\gamma_+ - \gamma_-)\infty\right)
$$
with Jacobian coming from
\ba\label{Jacobian}
d\xi = \left|\frac{\lambda_+ (\alpha_{11} - \lambda_- \beta_{11})}{\lambda_- (\alpha_{11} - \lambda_+ \beta_{11})}\right|^{1/2} \frac{\gamma_- - \gamma_+}{(t-\gamma_+)^2}dt \ ,
\ea
which shows that the integration can be actually done for $\xi\in(-\infty,\infty)$. Introducing
\ba\label{kappa}
\kappa^2 = \left|\frac{\alpha_{11} - \lambda_- \beta_{11}}{\alpha_{11} - \lambda_+ \beta_{11}}\right| \left|\frac{\lambda_- (\alpha_{11} - \lambda_+ \beta_{11})}{\lambda_+ (\alpha_{11} - \lambda_- \beta_{11})}\right| \equiv \frac{\lambda_-}{\lambda_+} < 1
\ea
and combining this with (\ref{AB1}), (\ref{xi}) and (\ref{Jacobian}) we get
\ba\label{final1}
\mathfrak{I} = \left| \frac{\lambda_+ - \lambda_-}{(\gamma_+ - \gamma_-)\sqrt{\lambda_+ |(\alpha_{11} - \lambda_+ \beta_{11})(\alpha_{11} - \lambda_- \beta_{11})|}} \right| \int\limits_{-\infty}^{\infty}\frac{d\xi}{\sqrt{(1+\xi^2)(1+\kappa^2 \xi^2)}}\ ,
\ea
which could be further simplified by using (\ref{nice4}) resulting in a completely ``geometric'' expression
\ba\label{final2}
\mathfrak{I} = \sqrt{\frac{\lambda_-}{\det\alpha}} \int\limits_{-\infty}^{\infty}\frac{d\xi}{\sqrt{(1+\xi^2)(1+\kappa^2 \xi^2)}}\ .
\ea
Using the standard change of variables, $\xi = \frac{1}{\kappa}\cot\theta$ we finally arrive at the result written in the Legendre canonical form (\ref{final3})
\ba
\mathfrak{I} = 2\sqrt{\frac{\lambda_-}{\det\alpha}} \int\limits_{0}^{\pi/2}\frac{d\theta}{\sqrt{1 - (1-\kappa^2) \sin^2\theta}} \equiv 2\sqrt{\frac{\lambda_-}{\det\alpha}} K\left(\sqrt{1-\frac{\lambda_-}{\lambda_+}}\right)\ ,
\ea
where the complete elliptic integral of the first kind is standardly defined as
\ba\label{ellipticK}
K(k)= \int\limits_{0}^{\pi/2}\frac{d\theta}{\sqrt{1 - k^2 \sin^2\theta}}\ .
\ea
\end{proof}

Though the result (\ref{final3}) does not seem to be symmetric in $\alpha$ and $\beta$, it actually is.
\begin{corollary}
The integral $\mathfrak{I}$ given in (\ref{final3}) is symmetric with respect to the exchange of $\alpha$ and $\beta$.
\end{corollary}
\begin{proof}
This is a trivial consequence of the properties of the symmetric polynomials:
\ba\label{sympoly}
\det\alpha\,\, e_1(\alpha^{-1}\beta) = \det\beta\,\, e_1(\beta^{-1}\alpha)\ ,\ \ \det\alpha\,\, e_2(\alpha^{-1}\beta) = \det\beta\ ,\ \ \det\beta\,\, e_2(\beta^{-1}\alpha) = \det\alpha\ ,
\ea
which could be easily established by, for example, comparing different expansions of $\det(\alpha - \lambda\beta)$ as in (\ref{lambda}). Then we have two equivalent expressions for $\lambda_\pm$, cf. (\ref{lambda})
\ba\label{lambda2}
&&\lambda_{\pm} = \frac{e_1(\alpha^{-1}\beta)\pm \sqrt{(e_1(\alpha^{-1}\beta))^2-4e_2(\alpha^{-1}\beta)}}{2e_2(\alpha^{-1}\beta)} \ ,\nonumber \\
&&\lambda_{\pm} = \frac{e_1(\beta^{-1}\alpha)\pm \sqrt{(e_1(\beta^{-1}\alpha))^2-4e_2(\beta^{-1}\alpha)}}{2} \ .
\ea
Then it is obvious that $\frac{\lambda_-}{\lambda_+}$ is explicitly symmetric in $\alpha$ and $\beta$. Using (\ref{sympoly}) and (\ref{lambda2}), it is easy to see that the same is true for $\frac{\lambda_-}{\det\alpha}$.
\end{proof}

Now, in the same way, we treat the integral (\ref{elliptic2})
\begin{proposition}\label{FFFintegral}
The integral (\ref{elliptic2}) is given by
\ba\label{final4}
\mathfrak{J} = 2\lambda_+\sqrt{\frac{\lambda_-}{\det\alpha}}\,E\left(\sqrt{1-\frac{\lambda_-}{\lambda_+}}\right) \equiv 2\sqrt{\frac{\lambda_+}{\det\beta}}\,E\left(\sqrt{1-\frac{\lambda_-}{\lambda_+}}\right)\ ,
\ea
where $E(k)$ is the complete elliptic integral of the second kind and $\lambda_{\pm}$ are as before.
\end{proposition}
\begin{proof}
The proof follows essentially the same steps as in Proposition \ref{FFintegral}, so we will use all the notations introduced there. Using (\ref{nice4}-\ref{kappa}) we easily obtain
\ba\label{final5}
\mathfrak{J} = \lambda_-\left| \frac{\lambda_+ - \lambda_-}{(\gamma_+ - \gamma_-)\sqrt{\lambda_+ |(\alpha_{11} - \lambda_+ \beta_{11})(\alpha_{11} - \lambda_- \beta_{11})|}} \right| \int\limits_{-\infty}^{\infty}d\xi\sqrt{\frac{1+\xi^2}{(1+\kappa^2 \xi^2)^3}}\ ,
\ea
which again by using (\ref{nice4}) can be brought to a ``geometric'' form
\ba\label{final6}
\mathfrak{J} = \lambda_-\sqrt{\frac{\lambda_-}{\det\alpha}} \int\limits_{-\infty}^{\infty}d\xi\sqrt{\frac{1+\xi^2}{(1+\kappa^2 \xi^2)^3}} \ .
\ea
Using the same change of variables, $\xi = \frac{1}{\kappa}\cot\theta$ we arrive at the result written in the Legendre canonical form (\ref{final4})
\ba
\mathfrak{J} = 2\lambda_+\sqrt{\frac{\lambda_-}{\det\alpha}} \int\limits_{0}^{\pi/2} d\theta \sqrt{1 - (1-\kappa^2) \sin^2\theta} \equiv 2\lambda_+\sqrt{\frac{\lambda_-}{\det\alpha}} E\left(\sqrt{1-\frac{\lambda_-}{\lambda_+}}\right)\ ,
\ea
where the complete elliptic integral of the second kind is standardly defined as
\ba\label{elliptic_E}
E(k)= \int\limits_{0}^{\pi/2} d\theta \sqrt{1 - k^2 \sin^2\theta}\ .
\ea
The last equality in (\ref{final4}) is easily established with the help of the relations (\ref{sympoly}).
\end{proof}

\bibliographystyle{utphys}
\bibliography{Ref}

\end{document}